\documentclass[12pt]{article}
\usepackage[affil-it]{authblk}
\usepackage[T1]{fontenc}
\usepackage{titling}
\usepackage[utf8]{inputenc}
\usepackage{setspace}
\usepackage{ulem}
\usepackage{amsmath,amsxtra,latexsym,amsthm,amssymb,amscd, mathtools}
\usepackage{nameref}
\usepackage{comment}
\usepackage[english]{babel}
\usepackage{subfig}
\usepackage{enumerate}  
\usepackage{graphicx}
\usepackage[margin=1.05in]{geometry}
\usepackage{mathrsfs}
\usepackage{graphics}
\usepackage{epsfig}
\usepackage{epic}
\usepackage{amsfonts}
\usepackage{scrextend}
\usepackage[colorinlistoftodos]{todonotes}
\usepackage[metapost]{mfpic}
\usepackage{tabularx}
\usepackage{relsize}
\usepackage{graphicx,natbib} 
\usepackage{blindtext}
\usepackage[colorlinks = true,
            linkcolor = blue,
            urlcolor  = blue,
            citecolor = blue,
            anchorcolor = blue]{hyperref} 
            
\usepackage{flexisym}
\theoremstyle{plain}
\newtheorem{theorem}{Theorem}
\newtheorem{proposition}{Proposition}

\newtheorem{lemma}{Lemma}

\newtheorem*{lemma*}{Lemma}

\theoremstyle{definition}

\newtheorem{example}{Example}

\newtheorem*{ifpart}{Sufficiency}
\newtheorem*{onlyifpart}{Necessity}

\makeatletter
\newcommand{\setword}[2]{%
  \phantomsection
  #1\def\@currentlabel{\unexpanded{#1}}\label{#2}%
}
\makeatother

\title{The probability of satisfying axioms: a non-binary perspective on economic design
}
\author{Pierre Bardier\thanks{pierre.bardier@ens.fr\newline I am deeply grateful to Marc Fleurbaey, Antonin Macé and William Thomson for their support, as well as for their detailed comments on this work. I thank the participants of the Workshop on Collective Decisions in Economic Analysis of the University of Alicante, the 9th International Workshop on Computational Social Choice in the University of Beersheba, the Conference on Economic Design in the University of Girona, the 16th meeting of the Society for Social Choice and Welfare in the Autonomous Technological Institute of Mexico, and the participants of the Online Social Choice and Welfare seminar, and of the economics seminar of the University of Caen.}}
\affil{Paris School of Economics and \'Ecole Normale Sup\'erieure de Paris,\\
48 Boulevard Jourdan 75014 Paris, France.}
\date{January 2025}

\begin{document}
\setlength{\droptitle}{-4em}
\onehalfspacing
\maketitle
\pagestyle{plain}

\vspace{-1.3cm}
\begin{abstract}

We provide a formal framework accounting for a widespread idea in the theory of economic design
: analytically established incompatibilities between given axioms should be qualified by the likelihood of their violation. We define the degree to which rules  satisfy an axiom, as well as several axioms, on the basis of a probability measure over the inputs of the rules.

Armed with this notion of degree, we propose and characterize: 

\begin{itemize}
    \item a criterion to evaluate and compare rules given a set of axioms, allowing the importance of each combination of axioms to differ, and
    \item a criterion to measure the compatibility between given axioms, building on a analogy with cooperative game theory.
\end{itemize}



\end{abstract}

\section{Introduction}\label{intro}

In the theory of economic design, 
incompatibilities between axioms have given rise to a myriad of notions of the \textit{degree to which a given axiom is satisfied}. However, these are, generally, model-and-axiom-specific. In contrast, this paper explores the potential of defining such a notion as the \textit{probability with which an axiom, as well as a set of axioms, is satisfied}, without restricting the analysis to particular types of properties or problems.

In the face of incompatibilities, notions of degree allow to compare in a nuanced way rules that are not comparable when sticking to the binary constraint according to which either a rule satisfies the axioms under consideration, if it meets the requirements for all the elements in its domain of definition, or does not, ``at all'', satisfy them. Take the example, without getting into details here, of ``non-dictatorial'' and ``non-trivial'' voting rules, defined on the universal domain of preferences associated with finite sets of alternatives and voters. A consequence of the Gibbard-Satterthwaite theorem is that the relative merit of any of these rules cannot be assessed on the basis of the full-fledged axiom of  ``strategy-proofness'' (\cite{gibbard1973manipulation}, \cite{satterthwaite1975strategy}). However, it is still possible, in principle, and useful, to compare the sensitivity to manipulation\footnote{Or the lack of ``strategy-proofness''.} of two ``non-dictatorial'' and ``non-trivial'' rules. One can, for instance, build an index representing the potential gains faced by voters misrepresenting their preferences in the two rules. Alternatively, one can compare the sets of preference profiles for which these rules are manipulable, using the partial order of set-inclusion. As this example suggests, two prominent interpretations support the use of notions of degree: one in which the parameters selected to measure the departure from a desirable property represent the \textit{intensity} of the violation, and one in which they represent its \textit{plausibility}. Our approach bears on the latter as we propose to compare rules according to the probability that they satisfy an axiom, or a set of axioms.\footnote{Other mathematical objects than (probability) measures can capture the plausibility of the violation. As an illustration, topological notions may be involved, for example, informally, in statements concluding that the set of preference profiles for which a rule does not meet the requirements of an axiom is ``Baire-negligible''.}

Studies discussing the likelihood that a rule satisfy a certain axiom, be that through empirical or theoretical analysis, all require that a specific way of \textit{counting the instances} for which the rule meets the considered requirements be chosen. These instances can be composed of theoretical preference profiles, or stochastically generated ones, sets of alternatives, parameters of actual elections, \textit{etc}. In that respect, simulation models, focusing initially mostly on the occurrence of the ``Condorcet paradox'' in voting (see \cite{geh1983} and \cite{gehlep2017} for reviews of this literature), are now commonly used in diverse settings (\textit{e.g.}, in addition to voting,  market design, fair division), under increasingly general statistical assumptions (\cite{wil2019}, \cite{disKam2020},  \cite{szufa2020drawing}, \cite{boehmer2021putting},  \cite{boehmer2023properties}, \cite{bohmputting}). 

In line with these models, while also accounting for other approaches (see Section \ref{litt}), we consider an \textit{abstract set of instances} ---the inputs of a rule--- endowed with a probability structure reflecting their relative frequency. The typical example of such a set in our view is the set of preference profiles associated with either a fixed or a varying group of agents. We can then measure the mass of instances for which not only ``punctual'' axioms, but also ``relational'' axioms (\cite{thomson2023axiomatics}) are verified.\footnote{We propose a formal definition of these two types of axioms in Section \ref{rulesaxioms}. Informally, some axioms are requirements made on outcomes obtained for each instance separately, while others formulate restrictions on outcomes obtained from different instances related in a specific way.} This general framework applies in any field in which an axiomatic approach is relevant, and in particular, covers a wide spectrum of market design, voting and social choice problems, be they Arrovian aggregation problems, voting problems, matching problems, fair division selection ---or ranking--- problems, with divisible or indivisible resources. Importantly, it does so while providing the degree of satisfaction of either single axioms or sets of axioms.

Defining the degree of satisfaction as a probability guarantees, \textit{in contrast to defining it on the basis of a notion of intensity}, its commensurability across (sets of) axioms. Concretely, the possibility to compare the extent to which a given rule satisfies two different combinations of axioms proves fundamental to \textbf{(1)} evaluate and compare rules, and \textbf{(2)} measure the compatibility of these axioms. 

Let us first illustrate the simple objects around which this work is structured. The questions raised in \textbf{(1)} and \textbf{(2)} will be addressed on the basis of collections of probabilities presented in arrays of the following form:
\begin{example}\label{First ex} Consider three axioms, $a_1,a_2$ and $a_3$:
\begin{equation*}
\begin{matrix}
 a_1 & a_2 & a_3 & a_1a_2 & a_1a_3 & a_2a_3 & a_1a_2a_3\\
 1 & 0.8 & 0.4 & 0.8 & 0.4 & 0.35 & 0.35
\end{matrix},
\end{equation*} where the 6th column, say, reads as ``the considered rule satisfies axiom $a_2$ and axiom $a_3$ \textit{simultaneously} with probability 0.35'' ---precise definitions are given in Section \ref{degree}.
\end{example}

Given a set of instances and a set of axioms, there is an intuitive criterion, to which we alluded in the discussion of the consequences of the Gibbard-Satterthwaite theorem, and with which the one we propose is consistent. According to it, a rule performs better than another one if, for each combination of axioms, the subset of instances for which it violates the requirements is included in the set of instances for which the other rule violates the requirements. Certainly, interesting comparisons of rules can be derived for some types of problems using this criterion (\cite{pathak2013school}, \cite{arribillaga2016comparing}, \cite{abdulkadiroǧlu2020efficiency}, \cite{abdulkadirouglu2021priority}).\footnote{A related approach consists in looking for a rule such that the set of instances for which the rule satisfies the axioms is maximal for inclusion (\cite{dasgupta2008robustness}, \cite{barbera2017sequential}).} In general, nevertheless, it induces a very partial ranking of rules: by working with a notion of degree based on probabilities of satisfaction, one obtains an \textit{extension} of such an order, accounting for the fact that some violations are more likely than others.\footnote{In the school choice context, for example, working with such a notion enables to take into account the correlation between students' preferences.}  

\textbf{(1)} \textit{Evaluating and comparing rules}. We actually consider a \textit{completion}\footnote{That is, an extension to a complete order.} of the partial order we just described. Indeed, we introduce and characterise a criterion to quantify the performance of a rule with regard to two key components. The first component is, as expected, probabilities of satisfaction. The second component focuses on the specific normative content of axioms. More precisely, the normative desirability of axioms, and, crucially, that of their combinations, are defined through the use of a \textit{capacity}.\footnote{A capacity is a real-valued function defined on the power set associated with the axioms, which gives value $0$ to the empty set, and is monotonic with respect to inclusion.} Not only can an axiom be more valuable to the eye of a researcher or a policy maker than another one, but \textit{synergies} are likely to emerge in the combination of axioms: conditionally on the satisfaction of a given axiom, the satisfaction of another one may be more or less valued, so that these axioms may be ``complementary'' or ``substitutable''. Capacities enable to capture this type of dependence. 

Importantly, this formulation has an operational interpretation. A decision maker, for example a policy maker in charge of selecting a mechanism to match students with schools, must distinguish between two rules on the basis of several principles that are logically incompatible. She then asks a team of researchers to estimate, for each rule, how probable the satisfaction of each combination of principles is. After determining it, the research team inquires about the relative importance of each combination of principles for the policy maker. The task is then to integrate these two pieces of information in order to decide on the rule to adopt.

Loosely speaking, we identify the only measure of performance that $i)$ consistently extends the natural measure for the case of a single axiom, while $ii)$ disentangling the probability with which a rule satisfies a set of axioms and the probability with which it satisfies a superset of it ---see Theorem \ref{theorem1}. The necessity to do so comes from the monotonicity of capacities: the valuation of the satisfaction of a given combination of axioms is incorporated in that of a superset of it. We show that a measure failing point \textit{ii)} displays some redundancy and may thus wrongly lead to the conclusion that some rule performs better than another one.

\textbf{(2)} \textit{Measuring the compatibility of axioms}. Finally, a collection of probabilities can be analysed in order to determine the degree of compatibility, or, equivalently, the degree of incompatibility, of axioms, given a rule, or given a family of rules. When such a collection is associated with one specific rule, computing how likely this rule is to satisfy a certain axiom, \textit{given that it satisfies some others}, enables to better understand its behaviour. One can also analyse collections of probabilities to identify how (in)compatible axioms are, given a domain of admissible rules.

We introduce and characterise a criterion fulfilling this purpose based on an analogy with cooperative game theory. Admissible collections of probabilities are naturally associated with  a unique cooperative game and, \textit{on the obtained restricted set of games}, we identify the Shapley value as the most adequate measure ---see Theorem \ref{theorem2}.

In Section \ref{litt}, we situate our approach in relation to the literature. In Section \ref{degree}, we provide general definitions of rules, of ``punctual'' and ``relational'' axioms, and of the degree to which a rule satisfies a given set of axioms. We address question \textbf{(1)} in Section \ref{performance}, after characterising the set of admissible collections of probabilities. We address question $\textbf{(2)}$ in Section \ref{axiomscomp}. Finally, in Section \ref{discuss}, we discuss methods to proceed to a robust analysis with respect to the probabilities of satisfaction.

\section{Related Literature}\label{litt}

\cite{Tho2001}, in a paper in which he seeks to characterise the essential features of the \textit{axiomatic program}, conceives this research as the attempt to draw as precise a frontier as possible between axioms that are compatible and axioms that are not. It is then possible to distinguish, for a given set of axioms, families of problems for which they can all be satisfied, and families for which they cannot. This view has motivated the most standard way of dealing with impossibilities in the theory of economic design: when some axioms are shown to be incompatible on a given domain of parameters, it seems natural to look for restricted domains in which these axioms can actually be combined. Accordingly, the plausibility of the compatibility of these axioms becomes the plausibility of the restricted domains, and it is left to the consumer of the theory to assess how suitable the domain restrictions are in the context at hand. Recently, this type of approach has saliently been described in \cite{Mou2019}, reviewing new developments in the theory of fair allocation, centered around very structured problems such as ones with ``one-dimensional single-peaked preferences'', ``dichotomous preferences'', or ``preferences with perfect substitutability''. Restricted preference domains such as those described above have also received special interest in algorithmic social choice theory, in particular because their simpler structure is likely to decrease the complexity of algorithms (\cite{BCELP16}).  

Yet, this approach maintains the binary constraint according to which a given condition is satisfied on a whole domain of parameters or is not, ``at all'', satisfied, whereas constructing a \textit{less partial order} between rules would require to know, when one fails to yield the desired outcomes, by how much it fails. For that matter, the use of parametrically weakened versions is quite classical: one or several parameters indicate the intensity of departure from the original studied property, see \cite{moutho1988}, \cite{schummer2004almost}, \cite{branal2011}, \cite{chevaleyre2017distributed} and \cite{skowron2021proportionality} for instance.\footnote{Actually \cite{moutho1988} show how parametric relaxations can be used to demonstrate the salience of the incompatibility of given principles. The number of papers introducing parametric relaxations is extremely large and this list is by no means exhaustive. All the cited papers belong to a different branch of the economic design literature. Parametric relaxations of axioms are also studied in individual decision theory: \cite{chambers2025decision} propose, for example, a way to express the degree to which a preference relation violates the classical ``independence axiom'' ---while it satisfies the other axioms of expected utility--- and establish a relation between this degree and the distance between a utility function representing the violating preference and a utility function representing a preference satisfying all axioms of expected utility.}

However, parametrizations are model-and-axiom-specific, which makes them, most often, incomparable to each other. In other words, most often, the definition of parametrized versions of two axioms gives no clue on how to define the degree to which they are simultaneously satisfied. In contrast, in this paper, we exploit the commensurability that a probability notion offers when measuring the performance of rules, as well as when studying the compatibility of axioms. 

Another theoretical approach, closer to the way we proceed, was adopted in the context of Arrovian social choice theory in \cite{campbell1994trade, campbell2015social}. The method is to ``count'', using a (probability) measure, the pairs, or triples, of alternatives for which studied axioms are satisfied in order to identify \textit{trade-offs} between them. One can describe this method in the general terms of our paper: the set of instances endowed with a measure structure in these papers is the set of pairs, or triples, of alternatives ---and not, for example, the set of preference profiles. For us, the point of  considering an abstract set of instances is precisely to be able to account for various approaches \textit{i)} involving different sets on which a measure is defined, and \textit{ii)} deriving degrees of satisfaction through different methods, \textit{e.g.} mathematical analysis, as in \cite{campbell1994trade, campbell2015social}, simulations or econometric estimations.

We already discussed in the introduction how this work relates to simulation models. Their general principle is to derive, from (statistical) assumptions on the behaviour and the preferences of agents involved in a given aggregation problem, the probability of occurrence of certain types of outcomes under different rules. A recent review of this vast literature can be found in \cite{disKam2020}. Let us mention a few examples of studies in voting and market design primarily consisting in measuring the empirical frequency of the violation of a given property, different from ``Condorcet consistency''.\footnote{Once again, see \cite{geh1983} and \cite{gehlep2017} for a review of the literature specifically dedicated to the ``Condorcet paradox''. See also \cite{lepelley2000computer} and \cite{laslier2010silico}.} In the former literature, \cite{brandt2014identifying} study the number of solutions selected by standard tournament solution concepts, using both real world preference data and simulations, thus testing for their (lack of) decisiveness. Focusing on the occurrence of the ``agenda contraction paradox'', \cite{brandt2016analyzing} conclude, based on simulations used to extend theoretical results obtained for problems involving four alternatives, that sensitivity to such contraction is of higher practical relevance than the ``Condorcet loser paradox''. \cite{aleskerov2012manipulability} study the level of manipulability of multi-valued rules, using computer experiments on problems with four and five alternatives, after extending theoretical indices used for single-valued social choice procedures. In market design, \cite{roth1999redesign} conduct simulations on data from the ``National
Resident Matching Program''
to account for the manipulability of the matching mechanism, and observed that even if it is in principle manipulable, the number of agents who would have an interest in returning a false report vanishes as the size of the market grows. \cite{ghasvareh2020fairness} (Chapter 4) compare the frequency of ``priority violations'' in three well-known many-to-one matching mechanisms that satisfy ``strategy-proofness'' and ``efficiency'', for different statistical distributions on preference profiles. 
\medskip

Taking stock, we believe that the present work can help analyse and compare rules in a subtle way by providing measures of performance that incorporate the probability to simultaneously satisfy several axioms, as well as their normative desirability and that of their combination. In particular, it provides a way to enrich the use of models based on notions of degree interpreted in terms of frequency of satisfaction.\footnote{\cite{wil2019} highlighted the importance of this issue for computer experiments.} 

\section{A commensurable notion of degree of satisfaction}\label{degree}

The core of our analysis is conducted on the basis of collections of probabilities such as the one in Example \ref{First ex}. The set of all collections is characterised in Section \ref{domainproba} and in the Appendix (Section \ref{Set of collections}). In this section, we propose a definition of axioms, and of the degree of satisfaction, which cover the vast majority of axioms studied in the theory of economic design.

 \subsection{Rules and axioms}\label{rulesaxioms}

Any notion of the frequency with which a given rule satisfies axioms requires considering \textit{instances} over which measuring its behaviour. As suggested above, letting preference profiles vary and analysing the \textit{outcomes} prescribed by a rule, which is often done in practice, is an obvious way to study instances ---and distributions over these instances. That is why, for concreteness, we use the case of varying preference profiles defined on a \textit{fixed set of alternatives} in all the illustrations of subsequent definitions. 

However, in order to be as general as possible, we introduce an abstract notion of instance\footnote{This enables, for example, to cover the case of varying ---and unequally likely--- sets of alternatives (\textit{e.g.}, depending on the considered type of problem, varying sets of candidates, varying budget sets, varying sets of schools, \textit{etc}.).}, from which \textit{classes of problems}, \textit{rules} and \textit{axioms} are defined. Informally, the frequency of satisfaction of an axiom will be defined as the measure of the set of instances for which the considered rule meets the stated requirements.

The starting point for evaluating the performance of rules is to specify the relevant domain: we define a \textbf{class of problems} as a pair of sets $(I,O)$, and refer to elements of $I$ as instances, and to elements of $O$ as outcomes.\footnote{All the objects we consider depend on a specific pair $(I,O)$, but this dependence is most often omitted in the following.} These objects respectively represent the arguments and the images of a rule: a \textbf{rule} $f$ is a mapping:
$$f: I \to O.$$

We stress that $O$ can have various structures, it may be a set of singletons of a given set ---so that $f$ is a function--- or a set of subsets of arbitrary size of a given set ---so that $f$ is a correspondence--- or a set of binary relations defined over a given set ---so that $f$ is a ranking mapping--- among other possibilities.

As an illustration, in the classical microeconomic division problem, fixing a social endowment $\Omega \in \mathbb{R}^l_+$ of $l \in \mathbb{N}$ divisible resources,\footnote{Given a set $B$ and a natural number $K$, $B^K$ denotes the $K-$fold Cartesian product of $B$. In addition, $\mathbb{R}^K_+$ ($\mathbb{R}^K_{++})$ denote the set of vectors in $\mathbb{R}^K$ with only non-negative (positive) components.} an instance is a profile of continuous, monotonic and convex individual preferences over $\mathbb{R}^l_+$ associated to a group $N$ of $n \in \mathbb{N}$ agents. A rule is then a correspondence, mapping each such instance $i$ to a set $o$ of vectors in $\mathbb{R}^{nl}_+$.

Before we define axioms, an important distinction should be made between, in the words of \cite{thomson2023axiomatics}, \textit{punctual} and \textit{relational} axioms. The former are requirements imposed on outcomes obtained for each instance separately, while the latter formulate restrictions on outcomes obtained from different instances related in a specific way. In the microeconomic division framework just mentioned, ``efficiency'' is a punctual axiom, and so is ``no-envy'', while ``population monotonicity'' is a relational one.\footnote{For definitions of these conditions, see, \textit{e.g.}, \cite{Mou2019} (Section 3.3).}

A \textbf{punctual axiom} $a$ is a mapping:
\begin{align*}
a: I &\to 2^O\\
i &\mapsto O_i^a.   
\end{align*}

In words, $a$ specifies for each instance a set of admissible outcomes, and the image of an instance under rule $f$ satisfies the requirements of $a$ if and only if it belongs to this set of admissible outcomes. Then, the typical exercise in economic design 
consists in finding a rule $f$ ---and, ideally, all rules $f$--- such that: $$\text{ for all } i \in I, \: f(i) \in O_i^a.$$

A \textbf{relational axiom} $a$ is associated with a parameter $K^a \in \mathbb{N}$ and is a mapping:
\begin{align*}
a: I^{K^a} &\to 2^{O^{K^a}}\\
(i_1,..., i_{K^a}) &\mapsto O^a_{i_1,...,i_{K^a}}.   
\end{align*}

Similarly, one typically looks for rules $f$ such that:

$$\text{ for all } ({i_1,...,i_{K^a}}) \in I^{K^a}, \: \big(f(i_1),...,f(i_{K^a})\big) \in O^a_{i_1,...,i_{K^a}}.$$ 

This is a general definition, but most relational axioms considered in different domains involve the comparison of outcomes obtained from only two different instances (\textit{i.e.} $K^a=2$ above). Returning to the above example, ``population monotonicity'' requires to consider, for a fixed social endowment, the outcomes of a rule when computed for a profile of preferences of a group of agents $N$ and a profile of a group $N \cup \{k\}$, $k \notin N$, which coincides with the preceding profile for agents in $N$. 

 In words, $a$ specifies for each tuple of instances a set of admissible tuples of outcomes, and the image of a tuple of instances under rule $f$ satisfies the requirements of $a$ if and only if it belongs to this set of admissible tuples of outcomes. There typically are many tuples of instances $(i_1,...,i_{K^a})$ for which $O^{a}_{i_1,...,i_{K^a}}=2^{O^{K^a}}$, that is, for which $a$ imposes no restriction whatsoever.\footnote{That is, a relational axiom associates with any tuple of instances \textit{related in a specific way} a \textit{specific} set of admissible tuples of outcomes, and does not associate with any other tuple of instances a restricted set of admissible tuples of outcomes.} In our example, ``population monotonicity'' is silent about pairs of preference profiles such that none is an extension of the other to a superset of agents.  

The reader can see that taking ${K^a}=1$ yields the definition of a punctual axiom; we however maintain this conceptually meaningful distinction for presentation purposes.\footnote{These definitions do not cover \textit{all} conceivable axioms. Some ``existential axioms'' (\cite{fishburn2015theory}), such as conditions on the range of a rule cannot be formulated in this way. They  still can be analysed in the way we propose in Sections \ref{performance} and \ref{axiomscomp}, by considering that they are satisfied by a rule either with degree $0$ or with degree $1$.

In addition, as noted in \cite{thomson2023axiomatics}, this distinction is \textit{soft} in the sense that some properties can be formulated both as punctual and as relational axioms. In such a case, the choice between these formulations is at the discretion of the researcher. \cite{schmidtlein2023voting} propose an interesting discussion on the different ways of defining an axiom.}

\subsection{The probability of satisfying axioms}\label{probability}

The vast majority of studies using stochastic preference models to generate instances, as well as the vast majority of papers involving a mathematical analysis of a set of instances endowed with a probability structure, have focused on \textit{i)} a single axiom at a time, and \textit{ii)} a punctual one. It is, however, possible to define the mass of instances for which a rule satisfies simultaneously several punctual or relational axioms.

Let $A$ be a finite set of $J \in \mathbb{N}$ axioms and $f: I \to O$ a rule. Collections of probabilities $(p^f_S)_{\emptyset \neq S \subseteq A} \in {[0,1]}^{2^{J-1}}$ such as the one given in Example \ref{First ex} are obtained in the following way. 

Let us first illustrate what the definition of the degrees of satisfaction would be if $A$ were only made of punctual axioms. Let $\xi$ be a $\sigma-$algebra  defined on $I$ and  $\mu$ a probability measure defined on $(I,\xi)$. Let $a$ be a punctual axiom, and  assume $D^f(a)=\{i \in I, 
f(i) \in O^{a}_i\}$, the set of instances whose image under $f$ meets the requirements imposed in $a$, is measurable.   
Then the degree to which $f$ satisfies $a$ according to $\mu$ is simply  $\mu\big(D^f(a)\big)$. Similarly, the degree to which all the axioms in $S \subseteq A$ are simultaneously satisfied is simply $\mu\big(\bigcap_{a\in S}\: D^f(a)\big)$.

The general definition covering the case of sets of relational axioms requires additional notation but its principle is the same. In order to cover both types of axioms, set $K^a=1$ when $a$ is punctual, even if this parameter is not needed for the definition of a punctual axiom. Similarly to what precedes, define $D^f(a)=\big\{(i_1,...,i_{K^{a}}) \in I^{K^a},
\big(f(i_1),...,f(i_{K^a})\big) \in O^{a}_{i_1,...,i_{K^a}}\big\}$ for $a \in A$, the set of tuples of instances whose images under $f$ meet the requirements imposed in $a$.

Let $K^{A}=\max_{a \in A} K^a$ and consider $I^{K^A}$, endowed with a $\sigma-$algebra  $\xi^{{K^A}}$. 
For a probability measure defined on $(I^{K^A},\xi^{K^A})$, denoted by $\mu$, the \textbf{degree} to which $f$ satisfies $S$, a non-empty subset of $A$, is given by:
\begin{equation*}
p^f_S=\mu\bigg(\bigg\{ \big(i_1,...,i_{K^{A}}\big) \text{ such that } \big(i_1,...,i_{K^a}\big) \in D^f(a) \text{ for all } a \in S \bigg\}\bigg).\footnote{We slighly abuse notation here by omitting the permutation used to restrict to the relevant $K^a$ instances for each $a \in S$. In addition, similarly to the case of a punctual axiom, $\bigg\{ \big(i_1,...,i_{K^{A}}\big) \text{ such that } \big(i_1,...,i_{K^a}\big) \in D^f(a) \text{ for } a \in S \bigg\}$ is assumed measurable in $(I^{K^A},\xi^{K^A})$.}
\end{equation*}

\textit{To summarise,} $p^f_S$ \textit{is the proportion ---computed from the probability measure} $\mu$\textit{--- of tuples of} $K^{A}$ \textit{instances such that, for any axiom} $a \in S$\textit{, the image under $f$ of their restriction to the} $K^a$ \textit{relevant instances satisfies the requirements of} $a$.\bigskip

\textbf{Remark:} For finite classes of problems, \textit{i.e.} for pairs $(I,O)$ such that $I$ and $O$ are finite, the measurability assumptions we introduced are innocuous. This is the case for voting problems with finitely many potential voters and candidates, and for allocation problems with finitely many potential agents and indivisible items. These two examples are of primary importance in our framework as they are studied in two fields of research where simulations are extensively used.

\section{How to measure the performance of rules ?}\label{performance} 

How can one assess the performance of a rule $f$ and, importantly, compare it with that of other rules, based on $p^f=(p^f_S)_{\emptyset \neq S\subseteq A}$, the probabilities of satisfying axioms in $A=\{a_1,...,a_J\}$, for which we proposed a definition in Section \ref{probability} ?

\subsection{Admissible collections of probabilities}\label{domainproba}

We address this issue by constructing a performance criterion defined for any $p$ in the set of possible collections of probabilities, a subset $P$ of $[0,1]^{2^{J-1}}$ characterised by consistency conditions relating, for all $\emptyset \neq S \subseteq A$, the probability to satisfy all subsets of $S$.

\textit{For instance}, the probability of satisfying all the axioms in $S$ cannot be larger than any of the probabilities of satisfying all of them but one, that is, to any probability in $(p_{S\setminus a})_{a\in S}$.\footnote{As is standard, we abuse notation by writing $S \setminus a$ rather than $S \setminus \{a\}$.} The probability of satisfying all axioms in $S$ is also constrained below. Select $a \in S$; taking $p_{S\setminus a}$ and $p_a$ as given, what is the worst case in terms of probability of satisfying $S\setminus a$ \textit{and} $a$? 
It corresponds to the situation in which the intersection of the sets of tuples of instances for which they are respectively satisfied has minimal measure, and the associated probability is $1-(1-p_{S\setminus a})-(1-p_a)=p_{S\setminus a}-(1-p_a)$ if it is positive, $0$ otherwise. In Example \ref{First ex}, given that the sets of tuples of instances for which $a_2$ and $a_3$ are satisfied have measure $0.8$ and $0.4$ respectively, the set for which they are simultaneously satisfied has at least measure $1-0.2-0.6=0.2$. 

We gave some necessary conditions for $p \in [0,1]^{2^J-1}$ to be an admissible collection of probabilities ---referred to as \textit{Fréchet inequalities}. They are not, nevertheless, sufficient, as the following example reveals. 

\begin{example}\label{Second Ex}
Let $A=\{a_1,a_2,a_3\}$ and consider:

\begin{equation*}
\begin{matrix}
&  & a_1 & a_2 & a_3 & a_1a_2 & a_1a_3 & a_2a_3 & A\\
& \mathbf{p} & 0.7 & 0.7 & 0.7 & 0.7 & 0.7 & 0.7 & 0.4
\end{matrix}.
\end{equation*}

The collection $p$ meets the two conditions above. However, letting $f$ be a rule associated with $p$, the set of tuples of instances for which $f$ satisfies $a_1$ and the set of tuples of instances for which $f$ satisfies $a_2$ are equal up to a set with measure $0$ ---if it were not the case, $f$ would not simultaneously satisfy the two axioms with the same probability as it satisfies each of them. Similarly, the set of tuples of instances for which $f$ satisfies $a_2$ and the set of tuples of instances for which $f$ satisfies $a_3$ are equal up to a set with measure $0$. In addition, all these sets have measure $0.7$. It is then inconsistent that $f$ simultaneously satisfy $a_1$, $a_2$ and $a_3$ with probability $0.4$ only: it must satisfy them with probability $0.7$.
\end{example}

The framework of \textit{probabilistic Boolean satisfyability} (\cite{nilsson1986probabilistic}, \cite{georgakopoulos1988probabilistic}) provides a natural formulation to identify necessary and sufficient conditions for a collection in $[0,1]^{2^J-1}$ to be a collection of probabilities. For the sake of brevity, we develop this idea in the appendix and simply state in this section the characterisation of $P$, the \textbf{set of possible collections of probabilities}.  

Before the statement, we need to introduce a specific family of collections. Let $\emptyset \neq S \subseteq A$. We define $p^{1,S}$ by:
\begin{align*}
    &p^{1,S}_{T}=1 \text{ if } \emptyset \neq T \subseteq S,\\
    &p^{1,S}_{T}=0 \text{ otherwise.}
\end{align*}

A rule with degrees of satisfaction given by  $p^{1,S}$ satisfies any combination of axioms which is not included in $S$ with probability $0$. In addition, it satisfies all the elements of $S$ with probability $1$, and \textit{thus}, all subsets of $S$ with probability $1$ ($S$ included). We comment more substantially on these collections later, as they correspond to special cases of the \textit{single-axiom-reducible problems} we introduce in Section \ref{charac}.

Let $\mathbf{0}$ denote the null vector in $\mathbb{R}^{2^J-1}$.\medskip

\begin{lemma}\label{ConvHull1}
$P$ is the closed convex hull of $\mathbf{0}$ and $(p^{1,S})_{\emptyset \neq S \subseteq A}$.    
\end{lemma}\medskip

The proof of Lemma \ref{ConvHull1}, as well as all subsequent proofs, is in the appendix.

Importantly, when defining a measure of performance \textit{for all} $p \in P \subseteq [0,1]^{2^J-1}$ (see Section \ref{charac}), we look for a measure that is applicable to \textit{any} set of $J$ axioms. For example, for two specific axioms $a_1$ and $a_2$ such that $a_1$ logically implies $a_2$ on the considered class of problems, it is impossible to find a probability measure on the set of instances such that $a_1$ is satisfied with a larger probability than $a_2$. Thus, a measure of performance that would be specifically tailored to $A=\{a_1,a_2\}$ would not need be defined for all consistent collections in $[0,1]^{3}$.

\subsection{Normative desirability of axioms and sets of axioms}\label{capacity}

As axioms most often reflect normative principles that matter to different extents, a key additional element for this evaluation needs to be introduced. More precisely, not only can an axiom be more valuable in the eyes of a researcher or a designer than another one, but \textit{synergies} are likely to emerge in the combination of axioms. For instance, in the microeconomic allocation framework, satisfying ``efficiency'' may be more or less valued than satisfying ``no-envy'', and, furthermore, the value of satisfying another fairness criterion such as ``egalitarian equivalence'', given the satisfaction of ``no-envy'', may be reduced, so that, equivalently, it may become more desirable to satisfy the efficiency condition. In this perspective, the example of a non-manipulability axiom such as ``strategy-proofness'' is also highly instructive. Indeed, the likelihood of truthful revelation is all the more important as other axioms involving requirements on preferences are satisfied\footnote{The satisfaction of ``strategy-proofness'' by itself may, of course, still be appreciated as, for example, it can be interpreted as preventing agents with lower ability to compute optimal actions from being disadvantaged (this interpretation has played an important role in the school choice literature (\textit{e.g.} \cite{artal2017}), but the interest of this axiom mainly lies in its interaction with other axioms.}, and, conversely, the satisfaction of these other axioms is all the more valuable that it is likely that they are applied to the actual ---truthfully revealed---  preferences.   

This observation leads us to allow the \textit{intrinsic valuation} of a non-empty combination $ S \subseteq A$ of the axioms to differ from the sum of the intrinsic valuations of axioms in $S$. As a consequence, we define the \textbf{set of possible intrinsic valuations} as the \textbf{set of capacities} on $A$:$$U=\bigg\{(u_S)_{\emptyset \neq S\subseteq A} \in \mathbb{R}_+^{2^{J}-1}, u_T \leq u_S \text{ if } T \subset S, \text{ for all } S \subseteq A \bigg\}.\footnote{\textit{We abuse language here as a capacity on set $A$, standardly, is also defined for the empty set, for which it returns value $0$.} The set $U$ would be appropriately referred to as a projection on $\mathbb{R}_+^{2^{J}-1}$ of the set of capacities. We however omit this qualification: an element of $U$ is called a capacity.}$$

By $U_{st}$ we denote the set of strictly monotonic capacities.\footnote{That is, $U_{st}=\bigg\{(u_S)_{\emptyset \neq S\subseteq A} \in \mathbb{R}_+^{2^{J}-1}, u_T < u_S \text{ if } T \subset S, \text{ for all } S \subseteq A \bigg\}.$}

The weak monotonicity assumption, with respect to inclusion, embedded in the use of capacities, can be interpreted as meaning that all axioms under consideration are normatively desirable. \textbf{Super-additive capacities} on $A$ are of special interest for our analysis. Let 
$$ U_{s.a}=\bigg\{ u \in U, u_S\geq u_T + u_{T'} \text{ if }  T \cup T' = S \text{ and }T \cap T'=\emptyset \text{ , for all } S \subseteq A\bigg\}. $$

A super-additive intrinsic valuation is interpreted as the result of \textit{complementarities} between all the considered axioms and is, for example, well suited to account for the interaction between ``strategy-proofness'', ``efficiency'' and ```no-envy'' as suggested above. From a general point of view, we see the use of super-additive valuations as the one most adequate to the typical problems studied in normative economics where the considered axioms are particular formulations of general and independent principles.\footnote{In the example we gave involving ``no-envy'' and ``egalitarian equivalence'', though, a super-additive valuation would not capture the effect we described.} 
We say that axioms in set ${A}$ are \textbf{complementary} when we consider a super-additive capacity on $A$.\footnote{Axioms in set $A$ are substitutes if the studied capacity is sub-additive. More generally, a capacity reflects complementarities between the axioms of set $T\subseteq A$ if for all disjoint 
$S,S' \subset T$, $u_T\geq u_S+u_{S'}$, and substitutability if the reverse inequality holds.}\bigskip

\textbf{Remark:} Clearly, we have a cardinal interpretation of intrinsic valuations: capacities express the intensity of preferences between combinations of axioms. This is key, for example, to capture complementarity through super-additivity. Note however that the rich information offered by the use of capacities in this framework might be partly, or completely, disregarded. One can restrict attention to capacities that only depend on the number of axioms that are satisfied: capacities of the form $u: S \in 2^A\setminus \{\emptyset\} \mapsto g(|S|) \in \mathbb{R}_+$, with $g: \mathbb{R}_+ \to \mathbb{R}_+$ non-decreasing.\footnote{For any set $B$, $|B|$ denotes the cardinality of $B$.} Such capacities may reflect the view that \textit{all} the axioms in $A$ are complementary, $g$ being convex, or that they all  are substitutes, $g$ being concave, but they imply neutrality across combinations of axioms of the same size. One could refer to these cases as \textit{uniform complementarity} and \textit{uniform substitutability}. \textit{Agnosticism} with respect to axioms can be captured through the use of the capacity $u:S \in 2^A\setminus \{\emptyset\} \mapsto |S| \in \mathbb{R}_+$.\footnote{The analysis would remain the same if we were to restrict attention to the set of normalised capacities, \textit{i.e.} the set of capacities such that $A$ is given value 1.} 

\subsection{Characterisation of the measure}\label{charac}

We have introduced the key ingredients for the construction of a performance measure, namely, the degrees of satisfaction, defined through probabilities, and the intrinsic valuations, defined through a capacity.

A \textbf{(performance) measure} is a mapping: $$m: U \times P \to \mathbb{R}_+,$$satisfying the following normalisation condition: for all $p \in P$, $$m(\mathbf{0},p)=0.$$ 

This natural property posits that the degrees to which a given rule satisfies the axioms in $A$ do not matter if these axioms are irrelevant to the decision maker ---note that a measure only takes non-negative values.

Given a pair of an intrinsic valuation and a collection of probabilities, the most intuitive measure for the performance of a rule associated with the collection arguably consists in taking the standard weighted sum: 
\begin{align*}
\Tilde{m}: U \times P &\to \mathbb{R}\\
(u,p) &\mapsto \sum_{\emptyset \neq S\subseteq A}\:u_Sp_S.   
\end{align*}

However, considering a capacity that positively depends on the cardinality of combinations, one can see that such a measure \textit{double counts} the satisfaction of some sets of axioms. 
\begin{example}\label{Fourth Ex} Let $A=\{a_1,a_2,a_3\}$ and consider:
\begin{equation*}
\begin{matrix}
& & a_1 & a_2 & a_3 & a_1a_2 & a_1a_3 & a_2a_3 & A\\
& \mathbf{u} & 1 & 1 & 1 & 3 & 3 & 3 & 6\\
& \mathbf{p}  & 0.7 & 0.8 & 0.5 & 0.7 & 0.25 & 0.3 & 0.25
\end{matrix}.
\end{equation*}

As $p_{a_1}=p_{a_1a_2}$, the rule associated with $p$ satisfies $a_1$ with exactly the same probability as it satisfies a combination of $a_1$ and another axiom while the intrinsic valuation of this combination incorporates the intrinsic valuation of $a_1$. Then, it is questionable that $a_1$ should have an impact on the measure under $p$, as is the case with $\Tilde{m}$. \textit{As a consequence, in our axiomatic approach, we will look for measures taking into account the difference between the probability with which a rule satisfies a given combination of axioms and the probability with which it satisfies any superset of it.} 
\end{example}

We first introduce a measure that would not double count the satisfaction of $a_1$ in Example \ref{Fourth Ex}. For all non-empty $S \subset A$, let $\hat{p}_S =\max_{T: S\subset T} \:\{p_T\}$ and let $\hat{p}_A=0$. The value $\hat{p}_S$ gives the maximal probability associated with a superset of $S$. Consider the following performance measure:
\begin{align*}
\hat{m}: U \times P &\to \mathbb{R}_+\\
(u,p) &\mapsto \sum_{\emptyset \neq S\subseteq A}\:u_S(p_S-\hat{p}_S).   
\end{align*}

We will call $\hat{m}$ the ``weighted minimal difference measure''. 

The redundancy we identified in the way the standard weighted sum $\Tilde{m}$ is computed is not merely a cardinal anomaly in the sense that it impacts the ranking of rules: in Example \ref{Fifth Ex} below, a rule associated with $p'$ performs better than a rule associated with $p$ according to the measure $\hat{m}$, which does not double-count in this example either\footnote{The value $\hat{m}(u,p)$ is computed as a weighted sum, so that a given set of axioms impacts this value if and only if it is given a non-null weight. Yet, in this example, all the sets which are satisfied with the same probability as one of their supersets are given weight $0$.}, while the standard weighted sum $\Tilde{m}$ yields the opposite conclusion. Note that $u$ only depends on the cardinality of combinations.
\begin{example}\label{Fifth Ex} Let $A=\{a_1,a_2,a_3\}$ and consider:
\begin{equation*}
\begin{matrix}
&  & a_1 & a_2 & a_3 & a_1a_2 & a_1a_3 & a_2a_3 & A\\
& \mathbf{u}  & 1 & 1 & 1 & 5 & 5 & 5 & 15\\
& \mathbf{p} & 0.7 & 0.7 & 0.7 & 0.6 & 0.6 & 0.6 & 0.6\\
& \mathbf{p'} & 1 & 1 & 0.45 & 1 & 0.45 & 0.45 & 0.45
\end{matrix}.
\end{equation*}

One has $\Tilde{m}(u,p)=20.1$, $\Tilde{m}(u,p')=18.7$, while $\hat{m}(u,p)=9.3$ and $\hat{m}(u,p')=9.5$.
\end{example}

We now identify another desirable property of a performance measure ---which is not satisfied by the standard weighted sum $\Tilde{m}$.

Let us place ourselves in the case where $A$ is a singleton. In this \textit{single-axiom case} where the considered rule satisfies axiom $a$ with probability $p_a$, while $a$ is given valuation $u_a$, it is fair to say that, given the cardinal interpretation of intrinsic valuations, the natural way to measure the performance of this rule is to take the ``expected valuation'' $p_au_a$. However, $\Tilde{m}$ does not constitute an appropriate generalisation to multiple axioms of this natural single-axiom measure.

Let $\lambda \in [0,1]$, $\emptyset \neq S \subseteq A$, and $p^{\lambda,S} \in {P}$ defined by:
\begin{align*}
    &p^{\lambda,S}_{T}=\lambda \text{ if } \emptyset \neq T \subseteq S,\\
    &p^{\lambda,S}_{T}=0 \text{ otherwise.}
\end{align*}

The collection $p^{\lambda,S}$ lies on the edge of $P$ between the extreme points $\mathbf{0}$ and $p^{1,S}$. Here is an illustration of $p=p^{0.6,a_1a_3}$, when $A=\{a_1,a_2,a_3\}$:
\begin{equation*}
\begin{matrix}
& & a_1 & a_2 & a_3 & a_1a_2 & a_1a_3 & a_2a_3 & A\\
& \mathbf{p} & 0.6 & 0 & 0.6 & 0 & 0.6 & 0 & 0
\end{matrix}.
\end{equation*}

A rule with degrees of satisfaction given by  $p^{\lambda,S}$ satisfies any combination of axioms which is not included in $S$ with probability $0$. In addition, it satisfies $S$ with exactly the same probability, $\lambda$, as it satisfies all the elements of $S$. As a consequence, we claim that the problem of measuring the performance of this rule, for a given intrinsic valuation $u$, has the same structure as a single-axiom problem in which the considered axiom is satisfied with probability $\lambda$ and is given valuation $u_S$. \textit{Hence, a measure that would provide a consistent generalisation to multiple axioms of the natural ``expected valuation'' defined for the single-axiom case, would, in contrast to} $\Tilde{m}$\textit{, return, for any} $(u,p^{\lambda,S}) \in U \times P$\textit{, the image} $p_Su_S=\lambda u_S$.\footnote{Of course, this already prevents some form of double-counting: when the sets of axioms that the rule satisfies with \textit{the same positive probability} form a chain, then, only the satisfaction of the maximal set of this chain should be taken into account in the measure of the performance.}\medskip 

While the observation that a measure should \textit{count once and only once} the satisfaction of a given subset of axioms requires some additional work in order to translate into a mathematical principle, this second observation immediately yields the following requirement:\\\bigskip

\textbf{\setword{Expected valuation for single-axiom-reducible problems}{Exp}}\medskip

Let $\lambda \in [0,1]$, $\emptyset \neq S \subseteq A$.

Let $u \in {U}$. Then, $$m({u},p^{\lambda,S})=\lambda u_S.$$\medskip

For example, the weighted minimal difference measure $\hat{m}$ satisfies this property. So does $(u,p) \mapsto \max_{\emptyset \neq S \subseteq A}\: u_Sp_S.$

\medskip

In the discussion of Example \ref{Fourth Ex} and the standard weighted sum $\Tilde{m}$, we identified an anomaly of double-counting: because there is a superset of $a_1$ that the rule satisfies with \textit{the same probability} as the one with which it satisfies $a_1$, there is no increment, in terms of probability, induced by focusing on the satisfaction of axiom $a_1$ only, and thus, $a_1$ should have no impact on the value $m(u,p)$. The fact that $P$ is $(2^J-1)-$dimensional convex polytope (Lemma \ref{ConvHull1}) is important in providing a definition of the increment in probability  associated with a non-empty set $S$ under collection $p$, covering the case in which $p_S>p_T$ for any superset $T$ of $S$.\bigskip

\textit{For all $p \in P\setminus \{\mathbf{0}\}$, there exists a} unique \textit{pair $\big(\mathcal{I}^p, (\alpha^p_T)_{\emptyset \neq T \subseteq A}\big)$,  where $\mathcal{I}^p$ is a non-empty family of non-empty subsets of $A$, and $(\alpha^p_T)_{\emptyset \neq T \subseteq A}$ a family of real numbers, such that:}

\begin{itemize}
    \item $0< \alpha^p_T  \leq 1$ \textit{for all} $T \in \mathcal{I}^p$ \textit{and} $\sum_{T \in \mathcal{I}^p} \alpha^p_T \leq 1$,
    \item $\alpha^p_T=0$ \textit{for all} $ \emptyset \neq T \in 2^A\setminus \mathcal{I}^p,$ \textit{and}
\end{itemize} 

$$p=\sum_{T \in \mathcal{I}^p} \alpha^p_T p^{1,T} \: \: \: \: \: \bigg(=\sum_{\emptyset \neq T \subseteq A} \alpha^p_T p^{1,T}+(1-\sum_{T \in \mathcal{I}^p} \alpha^p_T)\mathbf{0}\bigg).$$     
\bigskip

We write the trivial equality in parenthesis above in order to stress that, in general, the sum $\sum_{T \in \mathcal{I}^p} \alpha^p_T$ is not one. 

For the collection $\mathbf{0}$, we allow for the associated family of subsets to be empty and write $\mathcal{I}^{\mathbf{0}}=\emptyset$.

Importantly, there is a generic procedure enabling one to determine, for any $p \in P$, $\big(\mathcal{I}^p, (\alpha^p_T)_{\emptyset \neq T \subseteq A}\big)$. \textit{The computation of} $(\alpha^p_T)_{\emptyset \neq T \subseteq A}$ \textit{involves a recursive equation that makes clear that} $\alpha^p_T$ \textit{gives the increment we described above.} Let $p \in P$; we proceed inductively:

\begin{itemize}
    \item[] \textbf{Step 1.} If $p_A>0$, \text{ set } $\mathcal{I}_1^p=\{A\}$ and $\alpha^p_A=p_A$ , otherwise, set $\mathcal{I}_1^p=\emptyset$ and $\alpha^p_A=0;$
    \item[] \textbf{Step} $\mathbf{k}$ \textbf{(for} $\mathbf{2\leq k \leq J-1}$\textbf{).} Set $\mathcal{I}_k^p=\mathcal{I}^p_{k-1} \cup \{T \subseteq A \text{ with } |T|=J-k \text{ and } p_T-\sum_{S: T \subset S} \alpha^p_S > 0\}$, and, for all $T \subseteq A$ with $|T|=J-k$, $\alpha^p_T=p_T-\sum_{S: T \subset S} \alpha^p_S$.
    \item[] \textbf{Define} $\mathbf{\mathcal{I}^p=\mathcal{I}^p_{J-1}.}$
\end{itemize}

\begin{example}\label{Eighth} Let us illustrate the computation of $\alpha^p=(\alpha^p_T)_{\emptyset \neq T \subseteq A}$, with $A=\{a_1,a_2,a_3\}$:
\begin{equation*}
\begin{matrix}
& & {{a_1}} & a_2 & a_3 & {{a_1a_2}} & a_1a_3 & a_2a_3 & A\\
& \mathbf{p} & 0.7 & 0.8 & 0.5 & 0.7 & 0.25 & 0.3 & 0.25\\
& \boldsymbol{\alpha^p} & 0 & 0.05 & 0.2 & 0.45 & 0 & 0.05 & 0.25
\end{matrix}.
\end{equation*}    
\end{example}

Since $\alpha^p_A=p_A$ for all $p \in P$, one sees from the recursive equation $\alpha^p_T=p_T- \sum_{S:T\subset S}\alpha^p_S$ that $\alpha^p_T$ gives the increment in probability, induced by focusing on the satisfaction of axioms in $T$, rather than focusing on axioms in $T$ together with additional axioms. In other words, $\alpha^p_T$ gives the probability that the rule associated with $p$ satisfies all the axioms in $T$ and no other axiom. That is why we refer to $\alpha^p_T$ as the \textbf{contribution} of set $T$ under (the collection of probabilities) $p$.

Finally, this recursive definition implies that $\alpha^p$ is the solution of a \textit{Möbius inversion problem} (see Theorem \ref{theorem1} just below).

Motivated by the discussion of Example \ref{Fourth Ex}, we require that for all $p\in P$, and for all $u \in U$, the impact of $\emptyset \neq T \subseteq A$ on a measure $m$, \textit{given valuation} $u$, depend on the contribution of $T$ under $p$:\\\bigskip

\textbf{\setword{Same contribution---same impact}{Cont}}\medskip

Let $p,p' \in {P}$ and $\emptyset \neq S \subseteq {A}$ be such that $\alpha^p_S=\alpha^{p'}_S$.

Let $u\in {U}$ and $u^S \in U$ be such that $u^S_T={u}_T$ for all $T\neq S$. 

Then,  
$$ m(u^S,p)-m({u},p) = m(u^S,p')-m({u},p').$$\medskip

Valuations ${u}$ and $u^S$ above only differ, potentially, in their component associated with the combination $S$. Then, the differences above measure the impact of $S$ on measure $m$ under $p$ and the impact of $S$ under $p'$, respectively, \textit{given valuation} $u$\textit{, when one changes} $u_S$ \textit{to} $u^S_S$. We stress that this principle does not imply that the impact of an axiom under $p$ be the same for any $u \in U$.\medskip 

We are now able to present the characterisation of a unique performance measure, under the additional requirement that the projection of the measure on $U$ be a continuous mapping, where $U$ is endowed with the usual induced topology from $\mathbb{R}^{2^J-1}_+$.\footnote{For all $p \in P$, $m^p: u \in U \mapsto m(u,p) \in \mathbb{R}_+$ is a continuous mapping.} We simply write that the measure is continuous on $U$.\medskip

\begin{theorem}\label{theorem1}
A performance measure $m: U  \times P \to \mathbb{R}_+$, continuous on $U$, satisfies 
\begin{itemize}
    \item \ref{Cont}, and
    \item \ref{Exp}
\end{itemize}\medskip{}  if and only if it is the \textbf{weighted Möbius performance measure}, $\ddot m$:  

   \begin{align*}
    \ddot m : U \times P &\to \mathbb{R}_+\\
    (u,p) &\mapsto \sum_{\emptyset \neq S \subseteq A} u_S \bigg( \sum_{T: S \subseteq T}\: (-1)^{|T\setminus S|} p_T \bigg).
\end{align*}

\end{theorem}\medskip

The name of the characterised measure comes from the fact that the function associating $(p_S)_{\emptyset \neq S \subseteq A}$ with $\bigg(\sum_{T: S \subseteq T}\: (-1)^{|T\setminus S|} p_T\bigg)_{\emptyset \neq S \subseteq A}$ is the \textit{Möbius transform of the set function} $p: 2^A\setminus \emptyset \to [0,1]$ for the following partial order $\geq$ on $2^A\setminus \emptyset$:  $$\text{for all } S,T, \: \: S \geq T \iff S \subseteq T.$$

The reader is referred to the appendix and to \cite{grabisch2016set} (Chapter 2) for more details on the definition of the Möbius transform of a set function, given an arbitrary partial order defined on a (finite) set.\footnote{For all $(u,p) \in U \times P$, $\ddot m(u,p)= \sum_{\emptyset \neq S \subseteq A} p_S \bigg( \sum_{T\subseteq S}\: (-1)^{|S\setminus T|} u_T \bigg)$. The term in parenthesis corresponds to the Möbius transform, for the usual partial order defined by set inclusion, of a capacity $(u_S)_{S \subseteq A} \in \mathbb{R}^{2^J}$, returning $0$ for the empty set.}\bigskip

\textbf{Remark:} Let us insist on the precise role of the continuity assumption. If a performance measure $m: U \times P \to \mathbb{R}_+$ satisfies \textit{same contribution---same impact} and \textit{expected valuation for single-axiom-reducible problems}, then, for all $p \in P$, and all $u \in U_{st}$, $$m(u,p)=\sum_{\emptyset \neq S \subseteq A} u_S \bigg( \sum_{T: S \subseteq T}\: (-1)^{|T\setminus S|} p_T \bigg) +b^p, \text{ for some } b^p \in \mathbb{R},$$

where $b^{p^{\lambda,S}}=0$ for all $\lambda \in [0,1]$, and all non-empty $S \subseteq A$.\bigskip

We conclude this section with a comparison between the weighted Möbius performance measure $\ddot m$ and the measure $\hat{m}: (u,p) \mapsto \sum_{\emptyset \neq S\subseteq A}\:u_S(p_S-\hat{p}_S)$, illustrating in particular why $\ddot m$ properly accounts for the intrinsic valuation associated with any non-empty subset of $A$, while $\hat{m}$ does not. 
\begin{example}\label{Sixth Ex}
Let $A=\{a_1,a_2,a_3\}$; we let $w^{\hat{m}}$ and $w^{\ddot m}$ denote the collections of weights according to $\hat{m}$ and to $\ddot m$, respectively:
\begin{equation*}
\begin{matrix}
&  & a_1 & a_2 & a_3 & a_1a_2 & a_1a_3 & a_2a_3 & A\\
& \mathbf{p} & 0.55 & 0.6 & 0.2 & 0.35 & 0.05 & 0.15 & 0\\
& \mathbf{w^{\hat{m}}} & 0.2 & 0.25 & 0.05 & 0.35 & 0.05 & 0.15 & 0\\
& \mathbf{w^{\ddot m}} & 0.15 & 0.1 & 0 &0.35 & 0.05 & 0.15 & 0\\
\end{matrix}.
\end{equation*}

Let $u$ be a capacity defined on $A$. The two measures return the same weight for any subset of cardinality at least $2$. Let us focus on axiom $a_3$. In weighting $u_{a_3}$ by $0.2-0.05-0.15=0$, $\ddot m$ takes into account the fact that, by the monotonicity of $u$, in giving the weight $0.05$ to $u_{a_1a_3}$, and $0.15$ to $u_{a_2a_3}$, a portion $0.15+0.05$ of $u_{a_3}$ has already been incorporated in the measure of the performance of the rule. By focusing on the probability of satisfaction of a single superset of $a_3$ ($\hat{m}$ violates \textit{same contribution---same impact}), namely $a_2a_3$, $\hat{m}$ \textit{overweights} the valuation of $a_3$.\footnote{The reason why $\hat{m}$ did not overweight the valuation of $a_1$ in Example \ref{Fourth Ex}, and did not overweight the valuation of any subset for $p'$ in Example \ref{Fifth Ex}, is that for all these subsets, there existed a superset with the same probability of satisfaction: in such cases, the weighting formula of $\hat{m}$ equals that of the Möbius performance measure.}
\end{example}

\color{black}
\medskip

The way to measure the performance of a rule was one of the two natural questions raised when considering the collection of the degrees to which it satisfies all combinations of axioms. The second one pertains to measuring the (in)compatibility of axioms, given such a collection.

\section{Where does the incompatibility come from?}\label{axiomscomp}

A collection of probabilities associated with a specific rule may be used in order to analyse how compatible the considered principles are under this rule. Indeed, such a collection indicates how likely the rule is to satisfy a certain axiom, given that it satisfies any subset of other axioms, which enables to better understand its behaviour. 

Under this interpretation, as in the previous section, a collection $p \in P$ is associated with a single rule. We see however two additional ways to interpret such a collection. First, $p$ can be obtained as a  ``summary collection'' of multiple collections associated with different rules. As a simple illustration, $p$ may indicate the probability that \textit{all the rules} belonging to a given family satisfy the combinations of axioms in $A$.

Furthermore, in the case of several punctual axioms, one could inquire about the \textit{existence} of outcomes having all the required properties, rather than about the \textit{selection}, by a specific rule, of such an outcome. Using the notation, and following the reasoning, of Section \ref{degree}, this amounts to considering, for any non-empty $S \subseteq A$, $D(S)=\{i \in I \text{ such that } \bigcap_{a \in S} O_i^a \neq \emptyset\}$. Then, considering a $\sigma-$algebra $\xi$ on $I$, such that, for all $\emptyset \neq S \subseteq A$, $D(S)$ is measurable, and $\mu$ a probability measure on $(I,\xi)$, we define $p_S=\mu(D(S))$. The set of consistent collections of probabilities can then be described as in Section \ref{domainproba}.\footnote{Under the previous definition, in the appendix (Section \ref{Set of collections}), any $\emptyset \neq S \subseteq A$ is associated with a logical sentence whose truth value stands for the satisfaction, by a given rule, of all the axioms in $S$. Now, for each $\emptyset \neq S \subseteq A$, take the logical sentence $\Tilde{S}$ whose truth value stands for the non-emptyness of $\bigcap_{a\in S} O_i^a$; the analysis of Section \ref{Set of collections} goes through.} For instance, given a class of object allocation problems, a collection may represent the probability that Pareto efficient and envy-free allocations exist.

\textit{We address the question of how axioms interact, given a collection, without favouring any of these interpretations.} 

A first step towards the answer consists in building a measure of how compatible an axiom $a \in A$ is with the other axioms under $p$, and this can be done by determining how $a$ contributes to the \textbf{overall degree of incompatibility} $1-p_A$, compared to the other axioms. This question is akin to the general purpose of cooperative game theory where one tries to determine ways to allocate the benefits or costs of cooperation/interaction among a given set of agents.\footnote{Related to our approach is the literature focusing on the use of tools of cooperative game theory in feature attribution problems \big(\cite{lunlee2017}\big).}

In order to formalise this connection, define $$P^*=\bigg\{ \big(1,(p_S)_{\emptyset \neq S \subseteq A} \big), (p_S)_{\emptyset \neq S \subseteq A} \in P \bigg\}.$$

Each element of $P$ is associated with one and only one element of $P^*$; hence a generic element of $P^*$ will also be denoted by $p$ to alleviate notation, and we will sometimes write $p$ as $(p_S)_{S\subseteq A}$.\footnote{In words, each element of $P^*$ is obtained by choosing a unique consistent collection defined for all non-empty combination of axioms, and by adjoining to it the value $1$ for the empty-set, which should be though of as being always satisfied.} The notation $\mathbf{0}$ will now represent the null vector in $\mathbb{R}^{2^J}$.

Define $V=1-P^*=\{1-p, \: p \in P^*\}$, and let $v \in V$. Note that $v_{\emptyset}=1-p_{\emptyset}=0$. Thus, \textit{$v$ is a cooperative game associated with the set of axioms} $A$. The number $v_A=1-p_A$ represents the overall degree of violation of axioms in $A$ by a rule associated with $p$, and this magnitude must be distributed among them. 

All definitions below could be equivalently formulated as a requirement on $V$ or on $P^*$. In order to be consistent with the previous sections, we choose to write definitions on $P^*$.

An \textbf{incompatibility measure} is a mapping:
$$
    \psi: \: p \in P^* \mapsto \big(\psi_a(p)\big)_{a\in A} \in \mathbb{R}^J.$$

For all $a \in A$, $\psi_a$ gives a measure of the incompatibility of $a$ with the axioms in $A \setminus a$. 

The connection with cooperative game theory draws attention to specific incompatibility measures. For example, the \textbf{Banzhaf incompatibility measure} is defined by:

$$ 
\phi_a(p)=\sum_{S \subseteq A \setminus a} \frac{1}{2^{J-1}}(p_S-p_{S\cup a}),  \text{ for all } p \in P^*, \text{ and all } a \in A.$$

The \textbf{Shapley incompatibility measure} is defined by:

$$ \varphi_a(p)=\sum_{S \subseteq A \setminus a} \frac{|S| !(J-|S|-1) !}{J !}(p_S-p_{S\cup a}), \text{ for all } p \in P^*, \text{ and all } a \in A.$$

The Banzhaf and the Shapley measure differ in the weight associated with $p_S-p_{S\cup a}$. Let us briefly recall the standard interpretation of this weight for both measures. The Banzhaf measure for axiom $a$ weights the value $p_S-p_{S\cup a}$ by the probability that combination $S$ form in a scenario where all combinations are equally likely to form. The Shapley measure of axiom $a$ is constructed by weighting the value $p_S-p_{S\cup a}$ by the probability that the axioms in $S$, and only the axioms in $S$, arrive before $a$ in a scenario where axioms arrive one by one, according to a uniformly random permutation of $A$. 

For $a \in A$, and  $S \subseteq A \setminus a$, the difference $p_S-p_{S\cup a}$ interprets as the cost in probability of satisfaction that $a$ exerts on $S$. As a consequence, it is natural to require that the incompatibility measure associated with $a$ be a function of the cost exerted by $a$ on all $S \subseteq A\setminus a$:\\
\bigskip

\textbf{\setword{Same cost---same incompatibility}{Incomp}}\medskip

Let $p,p' \in {P}^*$ and $a \in {A}$ be such that, for all $S \subseteq A \setminus a$, $p_S-p_{S\cup a}=p'_S-p'_{S\cup a}$.

Then,  
$$ \psi_a(p)=\psi_a(p').$$
\medskip

This principle corresponds to the invariance property implied by Young's \textit{strong monotonicity} axiom in his classical characterisation of the Shapley value on the subspace of games associated with a fixed group of players (\cite{young1985monotonic}). The Banzhaf and the Shapley incompatibility measures satisfy \textit{same cost---same incompatibility}. According to both measures, the greater the value, the greater the incompatibility.

Two classical properties are necessary for $\psi$ to be a relevant measure in the problem we consider. The first one makes it possible to interpret $\psi$ as allocating the incompatibility $1-p_A$ among axioms in $A$.\\ 
\bigskip

\textbf{\setword{Allocation of incompatibility}{Allo}} 

Let $p \in P^*$; $$\sum_{a\in A}\psi_a(p)=1-p_A.$$\medskip 

This condition corresponds to the standard efficiency principle in cooperative game theory.

The last requirement states that the evaluation should not be biased towards any axiom:\\ \bigskip

\textbf{\setword{Anonymity}{Ano}} 

Let $p \in P$ and $\pi: A \to A$, a permutation. 

Let $p^{\pi}$ the collection defined by $p^{\pi}_S=p_{(\pi(a))_{a\in S}}$. 

Let $a \in A$. Then,
$$\psi_a(p)=\psi_{\pi(a)}\big(p^{\pi}\big).$$
\medskip

These three principles single out the Shapley incompatibility measure:\medskip

\begin{theorem}\label{theorem2} An incompatibility measure $\psi: P^* \to \mathbb{R}^J$ satisfies

\begin{itemize}
    \item \ref{Incomp},
    \item \ref{Allo}, and
    \item \ref{Ano}
\end{itemize} 
if and only if it coincides with the Shapley incompatibility measure: for all $a \in A$,
\begin{align*}
    \psi_a: P^* &\to \mathbb{R}^J\\
    p &\mapsto \sum_{S \subseteq A \setminus a} \frac{|S| !(J-|S|-1) !}{J !}(p_S-p_{S\cup a}).
\end{align*}  
\end{theorem}\bigskip

\textbf{Remark:} The reader will not be surprised that the Shapley incompatibility measure satisfies these three properties, given Young's axiomatisation of the Shapley value (\cite{young1985monotonic}) on the entire set of games $G=\{u=(u_S)_{S\subseteq A} \in \mathbb{R}^{2^J}, u_{\emptyset}=0\}$; however the fact that it is actually characterised by them is not immediate. Indeed, $P^*$ is defined by specific consistency conditions, and it turns out that the family of ``unanimity games'', used in his proof, does not belong to $V=1-P^*$. 

\begin{example}\label{Seventh Ex} Let $A=\{a_1,a_2,a_3\}$ and consider the ``unanimity game'' $u^{a_1a_3}$ defined by $u^{a_1a_3}_T=1$ if $a_1a_3 \subseteq T$, $u^{a_1a_3}_T=0$ otherwise. Then $p=1-u^{a_1a_3}$ is such that $p_{a_1}=p_{a_3}=1$ but $p_{a_1a_3}=0$, that is $p\notin P^*$, and thus $u^{a_1a_3} \notin V$. 
\end{example}

However, we build on the fact that each $v \in V$ can be written as a convex combination of points in $V$ ---the extreme points of $V$ are identified thanks to the proof of Lemma \ref{ConvHull1}--- and build on an induction argument that is analogous to the one that Young proposes.\footnote{The result holds if one replaces \textit{anonymity} by the following weaker well-known ``equal treatment" requirement: let $p \in P$, $a,a' \in A$ be such that, for all $S \subseteq A\setminus\{a,a'\}$, $p_S-p_{a \cup S}=p_S-p_{a' \cup S}$, then $\psi_a(p)=\psi_{a'}(p)$.}\medskip

The Shapley measure is also characterised   by \textit{allocation of incompatibility}, \textit{anonymity}, and direct adaptations of the classical ``null-player'' and ``additivity and positive homogeneity'' axioms ---the reader may find their explicit definitions in the appendix (see Theorem \ref{theorem3}). However, \textit{same cost---same incompatibility}, \textit{allocation of incompatibility} and \textit{anonymity} are in our view the very axioms that support the interpretation of $\psi$ as an adequate measure.

The following proposition draws a connection between the Shapley incompatibility measure and the Möbius transform that we used to build a measure of the performance of a rule.\bigskip

\begin{proposition}\label{Selectope}
    Let $a \in A$ and $p \in P^*$. Then, $$\varphi_a(p)=\sum_{S \subseteq A: a \notin S} \frac{\sum_{T: S \subseteq T} (-1)^{|T\setminus S|}p_T}{|A\setminus S |}.$$
\end{proposition}

The numerator of the summand is given by the Möbius transform of the set function $p: 2^A \to [0,1]$ for the following partial order $\geq$ on $2^A$: $$\text{for all } S,T, \: \: S \geq T \iff S \subseteq T.$$

The only difference with the previous section, due to the fact that we consider $p$ in $P^*$ and not in $P$, is that this transform must be defined for the empty set. In line with the previous section, let 
\begin{align*}
    \alpha^*: P^* &\to \mathbb{R}^{2^J}\\
    (p_S)_{S\subseteq A} &\mapsto \bigg(\sum_{T: S \subseteq T} (-1)^{|T\setminus S|}p_T\bigg)_{S \subseteq A}, 
\end{align*}
and $\alpha^{*p}$ denote $\alpha^*(p)$. That is, for all $\emptyset \neq S \subseteq A$, $\alpha^{*p}_S=\alpha^p_S$, and $\alpha^{*p}_{\emptyset}=1-\sum_{\emptyset \neq S \subseteq A}\alpha^p_S$.

For $S \subseteq A$ and $p \in P^*$, the value $\alpha^{*p}_S$ gives the probability of satisfying all the axioms in $S$ and only the axioms in $S$. Thus, from the point of view of an incompatibility measure, $\alpha^{p^*}_S$ represents what one loses by considering the other axioms. Then, the Shapley measure allocates this loss equally between these other axioms: each $a \in A \setminus S$ is allocated an incompatibility value of $\frac{\alpha^{*p}}{|A\setminus S|}$: for $a' \in A$, $\varphi_{a'}(p)=\sum_{S\subseteq A: a' \notin S} \frac{\alpha^{*p}}{|A\setminus S|}$.\footnote{A similar formula is known for the Shapley value on games. For all $u \in G$, all $a \in A$, the Shapley value associated with $a$ for game $u$ is given by $\sum_{S \subseteq A: a \in S}\frac{\sum_{T\subseteq S}(-1)^{|S\setminus T|}u_T}{|S|}$ (see \cite{grabisch2016set}, Chapter 3).}

\section{Discussion: robustness with respect to the collection of probabilities}\label{discuss}

One of the motivations for the use of a notion of degree defined as a probability of satisfaction was to extend a natural partial order between rules. According to it, a rule $f:I \to O$ performs better than a rule $f': I \to O$, when the set of axioms under consideration is $A$, if and only if, for all non-empty $S \subseteq A$, 
\begin{align}\label{inclusion}
\begin{split}
\: \: &\bigg\{ \big(i_1,...,i_{K^{A}}\big) \text{ such that } \big(i_1,...,i_{K^a}\big) \in D^{f'}(a) \text{ for all } a \in S \bigg\}\\
&\subseteq\\
&\bigg\{ \big(i_1,...,i_{K^{A}}\big) \text{ such that } \big(i_1,...,i_{K^a}\big) \in D^f(a) \text{ for all } a \in S \bigg\},
\end{split}
\end{align}
that is, for all possible combinations $S$, the set of (tuples of) instances for which $f'$ satisfies all the axioms in $S$ is included in the set of (tuples of) instances for which $f$ does. Allowing combinations of axioms to matter to a different extent in the evaluation of rules through the use of set functions that are monotonic with respect to inclusion, we have characterised a performance measure $\ddot m: U \times P \to \mathbb{R}_+$ inducing a \textit{completion} of this partial order. 

Fix $u \in U$ throughout this section, if Condition (\ref{inclusion}) holds for $f$ and $f'$, then, \textit{for any probability measure on} $(I^{K^A},\xi^{K^A})$ \textit{on the basis of which} $p^f$ \textit{and} $p^{f'}$ \textit{are computed}, $\ddot m(u,p^f) \geq \ddot m(u,p^{f'})$. However, when Condition (\ref{inclusion}) does not hold, there may of course be probability measures on $(I^{K^A},\xi^{K^A})$ for which the induced $p^f$ and $p^{f'}$ yield $\ddot m(u,p^f) \geq \ddot m(u,p^{f'})$ and others for which the induced $p^f$ and $p^{f'}$ yield $\ddot m(u,p^f) \leq \ddot m(u,p^{f'})$. This is a challenge when several such probability measures must be integrated into the analysis ---in particular when there is ambiguity on the measure that should be used to compute collections of probabilities. 

Two ways of dealing with such a variability stand out. For exposition purposes, we will describe them under the assumption that only finite sets of probability measures on $(I^{K^A},\xi^{K^A})$ are considered. A finite set of such probability measures is generically denoted by $\Delta=\{\mu_1,...,\mu_K\}$. Then, given such a set, a rule $f$ is associated with a finite set of collections of probabilities $\Pi^f_{\Delta}=\{p_1^f,...,p_K^f\}$.\footnote{One has $p^f_k=\bigg(\mu_k\bigg(\bigg\{ \big(i_1,...,i_{K^{A}}\big) \text{ such that } \big(i_1,...,i_{K^a}\big) \in D^f(a) \text{ for all } a \in S \bigg\}\bigg)\bigg)_{\emptyset \neq S \subseteq A}$, for all $k=1,...,K$.} The set of finite sets of collections is denoted by $\mathcal{P}$, and a generic element of it is denoted by $\Pi=\{p_1,...,p_H\}$.

The first way consists in \textit{selecting} a collection $p \in P$ supposed to represent or summarise the collections under consideration. Given a mapping $\kappa : \mathcal{P} \to P$, and given $\Delta$, one concludes that the rule $f$ performs better than the rule $f'$ if and only if $\ddot m\left(u,\kappa(\Pi^f_{\Delta})\right) \geq \ddot m\left(u,\kappa(\Pi^{f'}_{\Delta})\right)$. This order is always complete and $\kappa$ should be such that it extends the partial order defined by Condition (\ref{inclusion}). It is for instance the case with the following mapping: $$\kappa: \Pi=\{p_1,...,p_H\} \mapsto \sum_{h=1,...,H} \beta_h^{\Pi} p_h,$$ with $\beta_h^{\Pi}\geq0$ and $\sum_{h=1,...,H}\beta_h^{\Pi}=1.$

The other method consists in comparing two rules $f$ and $f'$, given $\Delta$, on the basis of the sets of values taken by the performance measure as $p^f$ and $p^{f'}$ vary in $\Pi^f_{\Delta}$ and $\Pi^{f'}_{\Delta}$, respectively. In that perspective, let us give several (not necessarily complete) standard criteria (see \cite{fishburn1985interval}, \cite{bewley2002knightian}, \cite{echenique2022twofold}, \cite{bardier2024unanimity}).\footnote{\cite{bewley2002knightian}, \cite{echenique2022twofold} and \cite{bardier2024unanimity} study incomplete criteria of decision making under ambiguity, and each of the three last criteria we describe in this section correspond to one proposed in one of these papers.}

$\boldsymbol{\alpha}$\textbf{-maxmin criterion:} Let $\alpha \in [0,1]$.
The rule $f$ performs better than the rule $f'$ if and only if:
\begin{align*}
\begin{split}
\: \: &\alpha \bigg(\max_{k=1,...,K} \ddot m(u,p^f_k)\bigg)+ (1-\alpha) \bigg(\min_{k=1,...,K} \ddot m(u,p^f_k)\bigg)\\
&\geq\\
&\alpha \bigg(\max_{k=1,...,K} \ddot m(u,p^{f'}_k)\bigg)+ (1-\alpha) \bigg(\min_{k=1,...,K} \ddot m(u,p^{f'}_k)\bigg).
\end{split}
\end{align*}
This criterion is complete and consistent with the partial order defined by Condition (\ref{inclusion}) for any value of $\alpha \in [0,1]$. When $\alpha=0$, the criterion focuses on a worst-case analysis, and, when $\alpha=1$, on a best-case analysis. 

All the criteria described so far induce a complete order between rules, but it is arguable that for comparisons to be robust, when facing different possible probability measures, one must account for the possibility that there be no ``sufficient evidence'' for two rules to be compared. The following criteria capture this idea.

\textbf{Max-and-min criterion:} The rule $f$ performs better than the rule $f'$ if and only if: 
\begin{align*}
    \begin{cases}
  \max_{k=1,...,K} \ddot m(u,p^f_k)\geq \max_{k=1,...,K} \ddot m(u,p^{f'}_k)\\
  \min_{k=1,...,K} \ddot m(u,p^f_k)\geq \min_{k=1,...,K} \ddot m(u,p^{f'}_k)
\end{cases}.
\end{align*}

\textbf{Point-wise criterion:} The rule $f$ performs better than the rule $f'$ if and only if: $$\text{ for all } k=1,...,K, \: \ddot m(u, p^{f}_k) \geq m(u, p^{f'}_k).$$

\textbf{Min-vs-max criterion:} The rule $f$ performs better than the rule $f'$ if and only if: 
$$\min_{k=1,...,K} \ddot m(u,p^f_k) \geq \max_{k=1,...,K} \ddot m(u,p^{f'}_k).$$

The \textit{max-and-min criterion} and the \textit{point-wise criterion} are consistent with the partial order defined by Condition (\ref{inclusion}). While extending this partial order was one of the main motivations of our approach, the \textit{min-vs-max criterion} is not consistent with it. More generally, there is a classical trade-off for these criteria between their degree of incompleteness and the conviction one can have in the comparisons they  express ---given $\Delta$, the \textit{max-and-min criterion} extends the \textit{point-wise criterion}, which extends the \textit{min-vs-max criterion}.\medskip

\textit{The primary insight from this section is that once we have identified, through our axiomatization, the measure to use for a single collection of probabilities, numerous standard methods, of which we simply  provided a few examples, can naturally be combined with the weighted Möbius performance measure in order to carry out a robust analysis.}

\section{Conclusion}\label{conc}

We defined a general notion for the degree to which a rule satisfies a set of axioms. Armed with it, we proposed and characterised (1) a unique criterion to evaluate the performance of rules, taking into account the normative desirability of axioms and, crucially, that of their combinations, and (2) a unique criterion to determine, for a given collection of probabilities of satisfaction, the role of each axiom in the overall degree of violation.

Let us further elaborate on the operational illustration we pointed out at the beginning of this work. A policy maker in charge of selecting a mechanism to match students with schools, must distinguish between (variants of) the ``Deferred Acceptance'' (DA), the ``Top Trading Cycle'' (TTC) and the ``Immediate Acceptance'' (IA) rules, on the basis of 
 ``strategy-proofness'' (for students), ``efficiency'' and ``stability'', which are incompatible, even in the standard case of strict preferences and priorities. She then asks a team of researchers to estimate, for each rule, how probable the satisfaction of each combination of principles is. The research team inquires about the relative importance of each combination of principles for the policy maker. Then, \textit{the weighted Möbius performance measure} adequately integrates these two pieces of information in order to measure and compare the performance of DA, TTC and IA. Each of these rules satisfies certain combinations of these properties with degree $1$. In addition, some constrained-optimality results, in the spirit of the partial order we mentioned throughout this work (see Condition (\ref{inclusion})), have been obtained ---we refer the reader to \cite{abdulkadirouglu2023school} for a review. However, the families of rules over which these results hold do not include all (variants of) the three rules, and thus do not provide a way to compare them.\footnote{For example, \cite{abdulkadirouglu2021priority} show that ``when each school has a single seat,
the top trading cycles algorithm has less priority violations than any Pareto efficient
and strategy-proof mechanisms.'' Neither DA or IA are among these rules.} For a fixed school choice environment, given an estimation of the preferences of students, and a valuation reflecting the importance of each combination of properties for the involved policy maker, one can compare these rules using our criterion.

In addition, given such an estimation of preferences, \textit{the Shapley incompatibility measure} enables to measure the compatibility of the three properties.

As we mentioned, certain notions of degree are defined to express the \textit{intensity} of the violation, rather than its \textit{plausibility}. A problem of comparability across axioms obviously emerges with such notions, and developing an analytical framework able to account for a certain \textit{partial commensurability}, inspired by the one we proposed here based on the \textit{complete commensurability} guaranteed by the use of probabilities, constitutes an important complementary research.


\section{Appendix}
\subsection{The set of collections of probabilities}\label{Set of collections}

With each non-empty set of axioms $S \subseteq A$, we associate a unique \textbf{logical sentence}, that we denote by $\Tilde{S}$, whose truth value stands for the satisfaction of all the axioms in $S$. There are thus $2^{J}-1$ such logical sentences. The set of sentences is denoted by $\mathcal{S}$. It induces a set of \textbf{possible worlds}, denoted by $\mathcal{W}$, defined as the set of \textit{logically consistent collections of truth values} of all sentences. That is, $\mathcal{W}$ is made of the $2^J$ collections of truth values $W=(w_{\Tilde{S}})_{\Tilde{S}\in \mathcal{S}}$, where $w_{\Tilde{S}} \in \{T,F\}$ ---$T$ standing for \textit{True}, $F$ for \textit{False}---
such that, for all $\emptyset \neq S \subseteq A$,
$$w_{\Tilde{S}}=T \text{ if and only if, for all } \emptyset \neq S' \subset S,  w_{\Tilde{S'}}=T.$$ 

One can illustrate this construction with a table where rows are sentences and columns are possible worlds. 
\begin{example}\label{Third Ex}
Let $A=\{a_1,a_2,a_3\},$ the set possible worlds $\mathcal{W}$ is represented by:

\begin{equation*}
\begin{matrix}
& & W_0 & W_1 & W_2 & W_3 & W_4 & W_5 & W_6 & W_8\\
& \Tilde{a_1} & F & T & T & T & T & F & F & F\\
& \Tilde{a_2} & F & F & T & F & T & T & T & F\\
& \Tilde{a_3} & F & F & F & T & T & F & T & T\\
& \Tilde{a_1a_2} & F & F & T & F & T & F & F & F\\
& \Tilde{a_1a_3} & F & F & F & T & T & F & F & F\\
& \Tilde{a_2a_3} & F & F & F & F & T & F & T & F\\
& \Tilde{A} & F & F & F & F & T & F & F & F\\
\end{matrix}.
\end{equation*}
\end{example}

Consider a collection $p=(p_S)_{\emptyset \neq S \subseteq A}$, with which we associate the collection $\Tilde{p}:\: \Tilde{S} \in \mathcal{S} \mapsto p_S \in [0,1]$, where $S$ is the non-empty subset of $A$ to which the logical sentence $\Tilde{S}$ corresponds in the construction above. In words, $p$ is consistent if there exists a probability measure $\pi$ on the set of possible worlds such that the value, according to $\Tilde{p}$, associated with any sentence, is equal to the sum of the values, according to $\pi$, associated with the possible worlds where the sentence has truth value $T$. 

Mathematically, take an arbitrary permutation of possible worlds and an arbitrary permutation of sentences, and define the incidence matrix $H$, of size $(2^{J}-1)\times 2^{J}$, by $H_{ij}=1$ if sentence $\Tilde{S}_i$ has truth value $T$ in world $W_j$, 0 otherwise. Let $\Delta^{\mathcal{W}}=\{\pi=(\pi_j)_{j=1,...,2^J}, 0 \leq \pi_j \leq 1, \sum_j \pi_j =1\}$.

We can now define the set of possible collections of probabilities, $P$. The collection $p=(p_S)_{\emptyset \neq S \subseteq A}$ belongs to $P$ if and only if the associated $\Tilde{p}$ is such that there exists $\pi \in \Delta^{\mathcal{W}}$ such that: $$\Tilde{p}=H\cdot \pi,$$
where $(\: \cdot\: )$ denotes the usual scalar product. 

The set $P$ is thus defined as the set of collections for which a specific linear system admits a solution. It remains to characterise the family of extreme points of $P$ (Lemma \ref{ConvHull1}):

$$P \text{ is the closed convex hull of } \mathbf{0} \text{ and } (p^{1,S})_{\emptyset \neq S \subseteq A};$$    
we prove it now.

\subsubsection*{\textbf{Proof of Lemma \ref{ConvHull1}}}

        The set of vectors $\Tilde{p} \in [0,1]^{2^J-1}$ for which there exists a solution in $\Delta^{\mathcal{W}}$ to the \textit{linear} system: $$\Tilde{p}=H\cdot \pi$$ is convex and closed. 

        Hence, $P$ is a convex compact subset of $[0,1]^{2^J-1}$, and, as such, it is the closed convex hull of its extreme points. It thus remains to prove that the extreme points of $P$ are $\mathbf{0}$ and $(p^{1,S})_{\emptyset \neq S \subseteq A}$.

        It is clear that these points are in $P$, and that they are extreme. These collections 
    are the only $\{0,1\}-$valued collections in $P$. Indeed, let $p$ be a $\{0,1\}-$valued collection that does not coincide with the null vector, or any $p^{1,S}$, $\emptyset \neq S \subseteq A$. Then, there exist $T,T'\subseteq A$ such that $\emptyset \neq T\subset T'$ and $p_{T'}=1$ and $p_T=0$; that is, $p \notin P$.

    Let $p \in P$ such that there is $ \emptyset \neq \Tilde{T}\subseteq A$ such that $0<p_{\Tilde{T}}<1$. Then there is $\epsilon >0$ such that the collections $\underline{p},\overline{p}$, defined by 
\begin{align*}
    &\underline{p}_T=p_T-\epsilon \text{ if } 0<p_T<1\\
    &\underline{p}_T=1 \text{ if } p_T=1\\
    &\underline{p}_T=0 \text{ if } p_T=0
\end{align*}
and 
\begin{align*}
    &\overline{p}_T=p_T+\epsilon \text{ if } 0<p_T<1\\
    &\overline{p}_T=1 \text{ if } p_T=1\\
    &\overline{p}_T=0 \text{ if } p_T=0
\end{align*}
are in $P$. In addition, $p=\frac{1}{2}\underline{p}+\frac{1}{2}\overline{p}$, that is, $p$ is not extreme. We have thus characterised the set of extreme points of $P$.\medskip

By the linear independance of the family $(p^{1,S})_{\emptyset \neq S \subseteq A}$, we have proved:\medskip

\textit{For all $p \in P\setminus \{\mathbf{0}\}$, there exists a} unique \textit{pair $\big(\mathcal{I}^p, (\alpha^p_T)_{\emptyset \neq T \subseteq A}\big)$,  where $\mathcal{I}^p$ is a non-empty family of non-empty subsets of $A$, and $(\alpha^p_T)_{\emptyset \neq T \subseteq A}$ a family of real numbers, such that:}

\begin{itemize}
    \item $0< \alpha^p_T  \leq 1$ \textit{for all} $T \in \mathcal{I}^p$ \textit{and} $\sum_{T \in \mathcal{I}^p} \alpha^p_T \leq 1$,
    \item $\alpha^p_T=0$ \textit{for all} $ \emptyset \neq T \in 2^A\setminus \mathcal{I}^p,$ \textit{and}
\end{itemize} 

$$p=\sum_{T \in \mathcal{I}^p} \alpha^p_T p^{1,T} \: \: \: \: \: \bigg(=\sum_{\emptyset \neq T \subseteq A} \alpha^p_T p^{1,T}+(1-\sum_{T \in \mathcal{I}^p} \alpha^p_T)\mathbf{0}\bigg).$$

We write the trivial equality in parenthesis above in order to stress that, in general, the sum $\sum_{T \in \mathcal{I}^p} \alpha^p_T$ is not one.

For the collection $\mathbf{0}$, we allow for the associated family of subsets to be empty and write $\mathcal{I}^{\mathbf{0}}=\emptyset$.

\subsection{Proof of Theorem \ref{theorem1}}

\begin{onlyifpart}
Let $m: U \times P \to \mathbb{R}_+$ be a \textit{continuous} measure on $U$ satisfying \textit{same contribution---same impact} and \textit{expected valuation for single-axiom-reducible problems}. For all $p \in P$, let $m^p: u \in U \mapsto m(u,p) \in \mathbb{R}_+$.\medskip

Let $u \in U_{st}$ and $\emptyset \neq S \subseteq A$. Consider $u^{S,\epsilon} \in U_{st}$ defined by $u^{S,\epsilon}_T=u_T$ if $T \neq S$ and $u^{S,\epsilon}_S=u_S+\epsilon$, for some $\epsilon>0$. There exists such a $u^{S,\epsilon}$ in $U_{st}$ because $u$ lies in $U_{st}$. 

Fix $p \in P$. For all $p' \in P$ such that $\alpha^p_S=\alpha^{p'}_S$, by \textit{same contribution---same impact}, 
    \begin{align*}
&m^p\bigg(\big(u_{a_1},\ldots,u_S+\epsilon,\ldots, u_A\big)\bigg)-m^p\bigg(\big(u_{a_1},\ldots,u_S,\ldots, u_A\big)\bigg)\\
=\: &m^{p'}\bigg(\big(u_{a_1},\ldots,u_S+\epsilon,\ldots, u_A\big)\bigg)-m^{p'}\bigg(\big(u_{a_1},\ldots,u_S,\ldots, u_A\big)\bigg).
    \end{align*}
    
    Yet, by definition of a single-axiom-reducible problem, for $\lambda=\alpha^p_S \in [0,1]$, 
$\alpha^{p^{\lambda,S}}_S=\lambda=\alpha^p_S$. Hence,
\begin{align*}
&m^p\bigg(\big(u_{a_1},\ldots,u_S+\epsilon,\ldots, u_A\big)\bigg)-m^p\bigg(\big(u_{a_1},\ldots,u_S,\ldots, u_A\big)\bigg)\\
=\: &m^{p^{\lambda,S}}\bigg(\big(u_{a_1},\ldots,u_S+\epsilon,\ldots, u_A\big)\bigg)-m^{p^{\lambda,S}}\bigg(\big(u_{a_1},\ldots,u_S,\ldots, u_A\big)\bigg)\\
=\: & \lambda(u_S+\epsilon)-\lambda u_S,
\end{align*} where the last equality follows from \textit{expected valuation for single-axiom-reducible problems}.\medskip

As a consequence, for all $u \in U_{st}$, all $p \in P$, and all non-empty $S \subseteq A$, $$\frac{m^{p}\bigg(\big(u_{a_1},\ldots,u_S+\epsilon,\ldots, u_A\big)\bigg)-m^{p}\bigg(\big(u_{a_1},\ldots,u_S,\ldots, u_A\big)\bigg)}{\epsilon}=\alpha^p_S.$$

By the convexity of $U_{st}$, the fact that $u^{S,\epsilon}$ lies in $U_{st}$ implies that for all $0<\gamma<\epsilon$, $u^{S,\gamma} \in U_{st}$.

One can thus let $\epsilon$ tend to $0$ in the $2^J-1$ equalities above, and this yields, as $U_{st}$ is open, that for all $p \in P$, $m^p$ is (continuously) differentiable on $U_{st}$ and its gradient vector is \textit{constant} and equal to $(\alpha^p_S)_{\emptyset \neq S \subseteq A}$. As $U_{st}$ is connected, we have proved that for all $p \in P$, there exists $b^p \in \mathbb{R}$ such that, for all $u \in U_{st}$: $$m^p(u)=\sum_{\emptyset\neq S \subseteq A}\: u_S\alpha^p_S + b^p.$$

As $m^p$ is continuous, and as $U$ is the closure of $U_{st}$, $m^p(u)=\sum_{\emptyset\neq S \subseteq A}\: u_S\alpha^p_S + b^p$ for all $u \in U$. We can now use $m(\mathbf{0},p)=0$ to conclude that for all $p \in P$, and $u \in U$

$$m(u,p)=\sum_{\emptyset \neq S \subseteq A} u_S \alpha_S^p.$$\medskip

It remains to give an explicit formula for $\alpha^p_S$. 

Recall that $\alpha^p_A=p_A$ for all $p \in P$.

For all $\emptyset \neq S \subset A$, and $p \in P$, the term $\alpha^p_S$ is defined recursively by 
\begin{align}
    &\: \: \alpha^p_S=p_S-\sum_{T: S \subset T} \alpha^p_T \notag \\
    \iff &\sum_{T: S \subseteq T} \alpha^p_T=p_S.\tag{$\ast$}
\end{align}\label{Mob}

Equation ($\ast$) corresponds to the formula defining the Möbius transform of $p$ associated with the partial order $\geq$ defined on $2^A\setminus \emptyset$, such that, $S \geq T \iff S \subseteq T$.\footnote{Again, we refer the reader to \cite{grabisch2016set} (Chapter 2) for more details on the definition of the Möbius transform of a set function, given an arbitrary partial order defined on a (finite) set.} The remaining of the proof is similar to the construction of the Möbius transform associated with a finite set partially ordered by inclusion in the standard way. 

By the classical result of Rota (\cite{rota1964foundations}), $(\alpha^p_S)_{\emptyset \neq S \subseteq A}$ satisfies Equation ($\ast$) for all non-empty $S \subseteq A$ if, and only if, there is a unique mapping $\nu: (2^{A}\setminus \emptyset) \times (2^{A}\setminus \emptyset) \to \mathbb{R}$ such that, 
\begin{itemize}
    \item $\nu(T,S)=1$ if $T=S$; $\nu(T,S)=-\sum_{S': S\subset S'\subseteq T}  \nu(T,S')$ if $S \subset T$, and $\nu(T,S)=0$ otherwise; and,
    \item for all non-empty $S \subseteq A$, $\alpha^p_S=\sum_{T: S\subseteq T} \nu(T,S) p_T.$
\end{itemize}

We prove by induction on the value of $|T\setminus S|$ that $\nu(T,S)=(-1)^{|T\setminus S|}$, for all $\emptyset \neq S,T$ with $S \subset T$.

If $|T\setminus S|=1$, then $\nu(T,S)=-\nu(T,T)=-1$. Assume the result holds for all $S',T'$ with $|T'\setminus S'|=k\in \mathbb{N}$ and let $S,T$ such that $|T\setminus S|=k+1$. Then,  
\begin{align*}
\nu(T,S)&=-\sum_{S': S\subset S'\subseteq T}  \nu(T,S')\\
&=-\sum_{S': S\subset S'\subseteq T}  (-1)^{|T\setminus S'|}\\
&=-\sum_{S': S\subseteq S'\subseteq T}  (-1)^{|T\setminus S'|}+(-1)^{|T\setminus S|}\\
&=(-1)^{|T\setminus S|}.
\end{align*}

The second equality follows from the induction hypothesis. The last one comes from the basic result in combinatorics, according to which $\sum_{S': S\subseteq S'\subseteq T}  (-1)^{|T\setminus S'|}$ is equal to $1$ if $S=T$ and to $0$ otherwise (see Lemma 1.1 in \cite{grabisch2016set}).

We have proved that $m$ coincides with the weighted Möbius performance measure $\ddot m$.\medskip

\end{onlyifpart}
\begin{ifpart}
    It is obvious that $\ddot m$ is continuous on $U$ and satisfies \textit{same contribution---same impact}.

    Let $\lambda \in [0,1]$, $\emptyset \neq S \subseteq A$, $p^{\lambda,S} \in {P}$, and $u \in {U}$. Then, 
    \begin{align*}
        \ddot m(u,p^{\lambda,S})=\sum_{\emptyset \neq B \subseteq A} u_B  \lambda\bigg( \sum_{T: B \subseteq T \subseteq S}\: (-1)^{|T\setminus B|}\bigg).
    \end{align*}

As, for all $B \subseteq A$, all $S \subseteq A$, 

\begin{align*}
        \sum_{T: B \subseteq T \subseteq S}\: (-1)^{|T\setminus B|}=\begin{cases}
            1, & \text{ if } B=S\\
            0, & \text{ otherwise}
        \end{cases},
    \end{align*}
 $m(u,p^{\lambda,S})=\lambda u_S$. We have proved that $\ddot m$ satisfies \textit{expected valuation for single-axiom-reducible problems}.
\end{ifpart}

\subsubsection*{\textbf{Independence}}

\begin{itemize}

\item Consider the mapping:\begin{align*}
    m: U \times P &\to \mathbb{R}\\
    (u,p) &\mapsto \max_{\emptyset \neq S \subseteq A} \: u_Sp_S.
\end{align*}

Such a performance measure satisfies \textit{expected valuation for single-axiom-reducible problems} but not \textit{same contribution---same impact}. Consider the following example, with $A=\{a_1,a_2,a_3\}$: 

\begin{equation*}
\begin{matrix}
&  & a_1 & a_2 & a_3 & a_1a_2 & a_1a_3 & a_2a_3 & A\\
& \mathbf{u}  & 1 & 1 & 1 & 3 & 3 & 2 & 6\\
& \mathbf{u^{a_1,a_2}}  & 1 & 1 & 1 & 3 & 3 & 5 & 6\\
& \mathbf{p} & 0.85 & 0.9 & 0.9 & 0.65 & 0.7 & 0.8 & 0.6\\
& \mathbf{p'} & 0.7 & 0.55 & 0.5 & 0.35 & 0.3 & 0.4 & 0.2
\end{matrix}.
\end{equation*}
\medskip

The reader can check that when considering $u$ and $u^{a_2a_3}$, the impact of $a_2a_3$ on $m$, as defined in the \textit{same contribution---same impact} principle is $4-3.6=0.4$ under $p$ while it is $2-0.2=0.8$ under $p'$ (and $\alpha^p_{a_2a_3}=\alpha^{p'}_{a_2a_3}=0.2$).

The measure $\hat{m}$ defined in Section \ref{charac} also satisfies \textit{expected valuation for single-axiom-reducible problems} but not \textit{same contribution---same impact}.\medskip

\item Consider the mapping:\begin{align*}
    m: U \times P &\to \mathbb{R}\\
    (u,p) &\mapsto \sum_{\emptyset \neq S \subseteq A} u_S\bigg(\sum_{T: S \subseteq T}\: (-1)^{|T\setminus S|} p_T \bigg)^2.
\end{align*}
Such a performance measure satisfies \textit{same contribution---same impact} but not \textit{expected valuation for single-axiom-reducible problems}: for all $\lambda \in [0,1]$, all $u \in U$, $m(u,p)=\lambda^2.$\medskip
\end{itemize}
Let us now display a performance measure which satisfies the two axioms, but is not continuous on $U$; $m:U \times P \to \mathbb{R}_+$: 
\begin{align*}
        m(u,p)=\begin{cases}
            \sum_{\emptyset \neq S \subseteq A} u_S \bigg( \sum_{T: S \subseteq T}\: (-1)^{|T\setminus S|} p_T \bigg) +b^p, \text{ for some } b^p \in \mathbb{R},&\text{if } u \in U_{st}\\
            \sum_{\emptyset \neq S \subseteq A} u_S \bigg( \sum_{T: S \subseteq T}\: (-1)^{|T\setminus S|} p_T \bigg) +4b^p & \text{ otherwise}
        \end{cases},
    \end{align*}
where, 
\begin{align*}
        b^p\begin{cases}
            >0 &\text{if } p \neq p^{\lambda, S} \text{ for any } \emptyset \neq S \subseteq A, \text{ and any } \lambda \in [0,1],\\
            =0 & \text{otherwise}
        \end{cases}.
\end{align*}

\subsection{Proof of Theorem \ref{theorem2}}

The \textit{sufficiency part} is readily checked. 

Suppose $\psi: P^*\to \mathbb{R}^J$ satisfies the three properties. Let
\begin{align*}
\Tilde{\psi}: V &\to \mathbb{R}^J\\
v &\mapsto \psi(1-v).\footnotemark
\end{align*}

\footnotetext{Briefly, \textit{same cost---same incompatibility} implies that for all $v,v'\in V$ and all $a \in A$ such that $v_{S \cup a}-v_S=v'_{S \cup a}-v'_S$ for all $S \subseteq A \setminus a$, $\Tilde{\psi}_a(v)=\Tilde{\psi}_a(v')$. Also, \textit{allocation of incompatibility} implies that for all $v \in V$, $\sum_{a \in A} \: \Tilde{\psi}_a(v)=v_A$. Finally, defining, for a permutation $\pi$, the game $v^{\pi}$ by $v^{\pi}_S=v_{(\pi(a))_{a\in S}}$, \textit{anonymity} implies that 
$\Tilde{\psi}_a(v)=\Tilde{\psi}_{\pi(a)}\big(v^{\pi}\big).$}

Consider $\varphi$ the Shapley incompatibility measure, and let 
$\Tilde{\varphi}$ denote the restriction of the classical Shapley value to the set of games $V$: for all $a \in A$, \begin{align*}
    \Tilde{\varphi}_a: V &\to \mathbb{R}\\
    v &\mapsto \varphi(1-v)=\sum_{S \subseteq A \setminus a} \frac{|S| !(J-|S|-1) !}{J !}(v_{S\cup a}-v_S).
\end{align*}  

Consider the family of games  $(\hat{v}^S)_{S\subseteq A}$ in $V$, where each game $\hat{v}^S$ is defined by  \begin{align*}
    \hat{v}^S_T=\begin{cases}
        1 &\text{ if } T \not \subseteq S\\
        0 & \text{ otherwise.}
    \end{cases} 
\end{align*}

Note that $\hat{v}^A=\mathbf{0}$, $\hat{v}^S=1-\hat{p}^{1,S}$ for all $\emptyset \neq S \subset A$, and $\hat{v}^{\emptyset}=\hat{p}^{1,A},$ where $\hat{p}^{1,S}$ denotes a collection of $\mathbb{R}^{2^J}$, such that $\hat{p}^{1,S}_{\emptyset}=1$, $\hat{p}^{1,S}_{T}=1$ if $\emptyset \neq T\subseteq S$, $0$ otherwise.

The following table illustrates this definition, with $A=\{a_1,a_2,a_3\}$, taking $v=\hat{v}^{a_1a_2}$ and the corresponding $p$:

\begin{equation*}
\begin{matrix}
& &\emptyset & a_1 & a_2 & a_3 & a_1a_2 & a_1a_3 & a_2a_3 & A\\
& \mathbf{v} & 0 & 0 & 0 & 1 & 0 & 1 & 1 & 1\\
& \mathbf{p} & 1 & 1 & 1 & 0 & 1 & 0 & 0 & 0
\end{matrix}.
\end{equation*}\medskip

We see from the proof of Lemma \ref{ConvHull1} that the set of extreme points of $V$ is the family $(\hat{v}^S)_{S\subseteq A}$ ---the set of extreme points of $P^*$ is the family $(\hat{p}^{1,S})_{S\subseteq A}$. Let $v \in V$. There exist a unique family of subsets of $A$, denoted by $\mathcal{I}^v$, and a unique family of positive real numbers $(\alpha^v_T)_{T\in \mathcal{I}^v}$, such that $\sum_{T \in I^v} \alpha^v_T =1$ and $$ v=\sum_{T \in \mathcal{I}^v} \alpha^v_T \hat{v}^T.$$

In particular, on $V\setminus \hat{v}^A$, for such families, the Shapley value associated with $a \in A$ is given by $$\Tilde{\varphi}_a(v)=\sum_{T \in \mathcal{I
}^v} \alpha^v_T \Tilde{\varphi}_a(\hat{v}^T)=\sum_{T \in \mathcal{I}^v: a \notin T} \alpha^v_T\frac{1}{|A\setminus T|},$$ and $\Tilde{\varphi}_a(\hat{v}^A)=0$. The reader may refer to Lemma \ref{Ext} below, where these equalities are proved. Moreover, it is not needed for this proof to explicitly give $(\mathcal{I}^v,(\alpha^v_T)_{T\in \mathcal{I}^v})$; this is simple though and we do so in the proof of Proposition \ref{Selectope} below. \medskip

We are now able to prove that the functions $\Tilde{\psi}$ and $\Tilde{\varphi}$ coincide on $V$. 

In the following, for a game $v \in V$, we say that two axioms $a,a' \in A$ are \textbf{symmetric} in $v$ if their transposition defines a symmetry of $v$. Formally, consider the permutation $\pi: A \to A$ defined by $\pi(\Tilde{a})=\Tilde{a}$ for all $\Tilde{a} \in A \setminus \{a,a'\}$, $\pi(a)=a'$ and $\pi(a')=a$. Axioms $a$ and $a'$ are symmetric in $v$ if, for all $S \subseteq A$, $v_{(\pi(a)_{a\in S})}=v_S$. \medskip

For $v \in V$, let $K^v$ denote the number of non-zero terms in a the convex combination of extreme points of $V$ to which $v$ is equal, described above.

If $K^v=0$, then $v=\mathbf{0}$ \textit{anonymity} and \textit{allocation of incompatibility} imply that, for all $a \in A$, $\Tilde{\psi}_a(v)=0=\Tilde{\varphi}_a(v)$. 

If $K^v=1$, then there is $T \subseteq A$ such that $v=\hat{v}^T$. If $T=A$, then $v=\mathbf{0}$ and one concludes as in the previous case. If $T=\emptyset$, by \textit{anonymity} and \textit{allocation of incompatibility}, for all $a \in A$, $\Tilde{\psi}_a=\frac{1}{J}=\Tilde{\varphi}_a$. Assume now $\emptyset \neq T\subset A$. For all $a \in T$, $v_{S\cup a}-v_{S}=0$ for all $S \subseteq A \setminus a$. Indeed, either $S \nsubseteq T$ and $v_{S\cup a}=v_{S}=1$, or $S \subset T$ and $v_{S\cup a}=v_{S}=0$. By \textit{same cost---same incompatibility}, this yields $\psi_a(v)=\psi_a(\mathbf{0})=0$. All $a,a'\notin T$ are symmetric in $v$ and we conclude, by \textit{anonymity} and \textit{allocation of incompatibility}, that $\Tilde{\psi}_a(v)=\Tilde{\psi}_{a'}(v)=\frac{1}{|A\setminus T|}=\Tilde{\varphi}_{a}(v)$.\medskip

We now proceed by induction on the value of $K^v$.\medskip

Assume that for all $v \in V \setminus \mathbf{0}$ such that $K^v\leq k \in \mathbb{N}$, for all $a \in A$,  $$\Tilde{\psi}_a(v)=\sum_{T \in \mathcal{I}^v: a \notin T} \alpha^v_T\frac{1}{|A\setminus T|}.$$ 

Let $v=\sum_{T \in \mathcal{I}^v} \alpha^v_T \hat{v}^T \in V\setminus \mathbf{0}$ with $K^{v}=k+1$. Consider $\mathcal{T}^v=\bigcup_{T \in \mathcal{I}^v} T$ and $a \in \mathcal{T}^v$. Define the game $$\nu=\sum_{T \in \mathcal{I}^v: a \notin T} \alpha^v_T \hat{v}^T + \big(1-\sum_{T \in \mathcal{I}^v: a \notin T} \alpha^v_T\big)\mathbf{0}.$$

Clearly, $\nu \in V$ and $K^{\nu} \leq k$. In addition, $\nu_{S\cup a}-\nu_{S}=v_{S\cup a}-v_{S}$ for any $S \subseteq A\setminus a$. Indeed, for all $S \subseteq A \setminus a$, for all $T \in \mathcal{I}^v$, 
$$\hat{v}^T_{S\cup a}-\hat{v}^T_{S}=
\begin{cases}
    0 &\text{ if } S \nsubseteq T\\
    0 &\text{ if } S \subset T \text{ and } a \in T \\
    1 &\text{ if } S \subseteq T \text{ and } a \notin T,  
\end{cases}$$
which implies: 
$$v_{S\cup a}-v_{S}=\sum_{\mathclap{\substack{T \in \mathcal{I}^v:
S \subseteq T,
a \notin T}}} \; \alpha^v_T=\nu_{S\cup a}-\nu_{S}.$$

Therefore, if $\nu \neq \mathbf{0}$, \textit{i.e.} if $a \notin \bigcap_{T \in \mathcal{I}^v} T$, $$\Tilde{\psi}_a(v)=\Tilde{\psi}_a(\nu)=\sum_{T\in \mathcal{I}^v: a \notin T} \alpha^v_T\frac{1}{|A\setminus T|}=\Tilde{\varphi}_a(v),$$
where the first equality follows from \textit{same cost---same incompatibility}, and the second follows from the induction hypothesis. And if $\nu=\mathbf{0}$, \textit{i.e.} if $a \in \bigcap_{T \in \mathcal{I}^v} T$, $\Tilde{\psi}_a(v)=0=\Tilde{\varphi}_a(v)$.

Moreover, all axioms in $A\setminus \mathcal{T}^v$ are symmetric in $v$.\footnote{For all $a\in A\setminus \mathcal{T}^v$, for all $S \subseteq A \setminus a$, $v_{S \cup a}=1$, thus, any transposition of two elements of $A\setminus \mathcal{T}^v$ is an symmetry of $v$.} As $\Tilde{\psi}_a$ coincides with $\Tilde{\varphi}_a$ for $a \in \mathcal{T}^v$, \textit{anonymity} and \textit{allocation of incompatibility} imply that for $a \notin \mathcal{T}^v$, $\Tilde{\psi}_a(v)=\Tilde{\varphi}_a(v)$.

We have proved that $\Tilde{\psi}$ coincides with $\Tilde{\varphi}$, which implies that the incompatibility measure $\psi$ coincides with the Shapley incompatibility measure.

\subsubsection*{\textbf{Independence}}
That the three principles, expressed for $\Tilde{\psi}: V \to \mathbb{R}$, that is, expressed for the restriction to $V$ of a \textit{solution}, are independent is shown in exactly the same way as when the set of admissible games is $G=\{u=(u_S)_{S\subseteq A} \in \mathbb{R}^{2^J}, u_{\emptyset}=0\}$. 

\subsection{Alternative characterisation of the Shapley incompatibility measure}
Consider the following principles.\medskip

\textbf{No cost---no incompatibility}\medskip 

Let $p \in P^*$, and $a \in A$ such that $p_{S\cup a}=p_a$ for all $S \subseteq A\setminus a$. 

Then, $$\psi_a(p)=0.$$\medskip

Such an axiom simply exerts no cost in terms of probability of satisfaction and should thus be considered as maximally compatible with the others.\medskip 

\textbf{Convex linearity}\medskip

Let $p,p' \in P^*$, $\lambda \in [0,1]$. 

Then, $$\psi_a(\lambda p+(1-\lambda) p')=\lambda \psi_a(p)+(1-\lambda) \psi_a(p').$$\medskip

This is a weakening of the classical \textit{additivity and positive homogeneity} principle required on the whole subspace $\left\{p \in \mathbb{R}^{2^J}, p_{\emptyset}=1\right\}$.\footnote{For completeness, let us state it: let $p, p^{\prime} \in \left\{p \in \mathbb{R}^{2^J}, p_{\emptyset}=1\right\}, \lambda \geq 0$.
Then,

$$
\psi_a\left(p+\lambda p^{\prime}\right)=\psi_a(p)+\lambda \psi_a\left(p^{\prime}\right).
$$
} It is suited for $P^*$, which is a compact and convex subset of this subspace. It is best interpreted as a simplicity requirement. 

\begin{theorem}\label{theorem3}
    An incompatibility measure $\psi: P^* \to \mathbb{R}^J$ satisfies

\begin{itemize}
    \item \textit{Convex linearity},
    \item \textit{Allocation of incompatibility}, 
    \item \textit{Anonymity}, and
    \item \textit{No cost---no incompatibility}
\end{itemize} 
if and only if it coincides with the Shapley incompatibility measure.
\end{theorem}
\begin{proof}
The \textit{sufficiency} is readily checked.

Suppose $\psi: P^* \to \mathbb{R}^J$ satisfies these four properties.
\begin{lemma}\label{Ext} Let $S\subseteq A$. Then, for all $a \in A$, $\Tilde{\psi}_a(\hat{v}^S)=\Tilde{\varphi}_a(\hat{v}^S)$.
\end{lemma}
\begin{proof}
If $S=A$, then $\hat{v}^S=\mathbf{0}$ and \textit{anonymity} and \textit{allocation of incompatibility} imply $\Tilde{\psi}_a(\hat{v}^S)=0=\Tilde{\varphi}_a(\hat{v}^S)$, for all $a \in A$. Let $S \subset A$ and consider $\hat{v}^S$. Let $a\in S$, then, by \textit{no cost---no incompatibility}, $\Tilde{\psi}_a(\hat{v}^S)=0=\Tilde{\varphi}_a(\hat{v}^S)$. In addition, all axioms in $A \setminus S$ are symmetric so that, by \textit{anonymity} and \textit{allocation of incompatibility}, for all $a \in A \setminus S$, $\Tilde{\psi}_a(\hat{v}^S)=\Tilde{\varphi}_a(\hat{v}^S)=\frac{1}{|A\setminus S|}$.
\end{proof}\medskip

By \textit{convex linearity}, for all $v \in V$, there exist a unique family of subsets of $A$, denoted by $\mathcal{I}^v$, and a unique family of positive real numbers $(\alpha^v_T)_{T\in \mathcal{I}^v}$, such that $\sum_{T \in \mathcal{I}^v} \alpha^v_T =1$ and $$ \Tilde{\psi}_a(v)=\sum_{T \in \mathcal{I}^v} \alpha^v_T \Tilde{\psi}_a(\hat{v}^T) \text{ for all } a \in A.$$

Then, by Lemma \ref{Ext}, $$ \Tilde{\psi}_a(v)=\sum_{T \in \mathcal{I}^v} \alpha^v_T \Tilde{\varphi}_a(\hat{v}^T)=\Tilde{\varphi}_a(v) \text{ for all } a \in A.$$

\end{proof}
That the four principles, expressed for $\Tilde{\psi}: V \to \mathbb{R}$, that is, expressed for the restriction to $V$ of a \textit{solution}, are independent is shown in exactly the same way as when the set of admissible games is $G=\{u=(u_S)_{S\subseteq A} \in \mathbb{R}^{2^J}, u_{\emptyset}=0\}$.

\subsection{Proof of Proposition \ref{Selectope}}

As we noted in the proof of Theorem \ref{theorem2}, the set of extreme points of the convex polytope $P^*$ is the family $(\hat{p}^{1,S})_{S \subseteq A}$ defined by $p^{1,S}_T=1$ if $T \subseteq S$ and $0$ otherwise, for all $T \subseteq S$.  For all $p \in P^*$, there exists a unique pair $\big(\mathcal{I}^{*p}, (\alpha^{*p}_T)_{T \subseteq A}\big)$,  where $\mathcal{I}^{*p}$ is a non-empty family of subsets of $A$, and $(\alpha^{*p}_T)_{T \subseteq A}$ is a family of real numbers such that

\begin{itemize}
    \item $0< \alpha^{*p}_T  \leq 1$ for all $T \in \mathcal{I}^{*p}$ and $\sum_{T \in \mathcal{I}^{*p}} \alpha^{*p}_T = 1$,
    \item $\alpha^{*p}_T=0$ for all $T \in 2^A\setminus \mathcal{I}^{*p},$ and
    $$p=\sum_{T \in \mathcal{I}^{*p}} \alpha^{*p}_T \hat{p}^{1,T}.$$
\end{itemize}

Let $p \in P^*$. Similarly to the procedure described in Section \ref{performance}, $\big(\mathcal{I}^{*p}, (\alpha^{*p}_T)_{T \subseteq A}\big)$ obtains as follows: 

\begin{itemize}
    \item[] \textbf{Step 1.} If $p_A>0$, \text{ set } $\mathcal{I}_1^{*p}=\{A\}$ and $\alpha^{*p}_A=p_A$ , otherwise, set $\mathcal{I}_1^{*p}=\emptyset$ and $\alpha^{*p}_A=0;$
    \item[] \textbf{Step} $\mathbf{k}$ \textbf{(for} $\mathbf{2\leq k \leq J}$\textbf{).} Set $\mathcal{I}_k^{*p}=\mathcal{I}^{*p}_{k-1} \cup \{T \subseteq A \text{ with } |T|=J-k \text{ and } p_T-\sum_{S: T \subset S} \alpha^{*p}_S > 0\}$, and, for all $T \subseteq A$ with $|T|=J-k$, $\alpha^{*p}_T=p_T-\sum_{S: T \subset S} \alpha^{*p}_S$.
    \item[] \textbf{Define} $\mathbf{\mathcal{I}^{*p}=\mathcal{I}^{*p}_{J}.}$
\end{itemize}

Then, $(\alpha^{*p}_T)_{T \subseteq A}$ is the Möbius transform of $p$ for the partial order $\geq$ defined on $2^A$ such that:

$$\text{for all } S,T, \: S \geq T \iff S \subseteq T.$$

Hence, for all $p \in P^*$, for all $S \subseteq A$, $\alpha^{*p}_S=\sum_{T:S\subseteq T}(-1)^{|T\setminus S|}p_T.$

The set of extreme points of $V$ is the family $(\hat{v}^S)_{S \subseteq A}=(1-\hat{p}^{1,S})_{S\subseteq A}$. Let $v=1-p$, $p \in P^*$, it is easy to see that $v=\sum_{S \in \mathcal{I}^{*p}} \alpha^{*p}_S \hat{v}^S$. Indeed, for all $T \subseteq A$:
\begin{align*}
    & p_T=\sum_{S \in \mathcal{I}^{*p}} \alpha^{*p}_S\hat{p}^{1,S}_T\\
    \iff &v_T=1-\sum_{S \in \mathcal{I}^{*p}} \alpha^{*p}_S\hat{p}^{1,S}_T\\
    \iff &v_T=1-\sum_{S \in \mathcal{I}^{*p}: T \subseteq S} \alpha^{*p}_S\\
    \iff &v_T=\sum_{S \in \mathcal{I}^{*p}} \alpha^{*p}_S \hat{v}^S_T.
\end{align*}

The last equivalence comes from the fact that $\sum_{S \in \mathcal{I}^{*p}} \alpha^{*p}_S \hat{v}^S_T=\sum_{S \in \mathcal{I}^{*p}: T \nsubseteq S} \alpha^{*p}_S$, and that $\sum_{S \in \mathcal{I}^{*p}: T \nsubseteq S} \alpha^{*p}_S+\sum_{S \in \mathcal{I}^{*p}: T \subseteq S} \alpha^{*p}_S=1.$

Then, for all $p \in P^*$, all $a\in A$, $$\varphi_a(p)=\Tilde{\varphi}_a(1-p)=\sum_{S \in \mathcal{I}^{*p}}\alpha^{*p}_S \Tilde{\varphi}_a(\hat{v}^S)=\sum_{S \subseteq A: a \notin S}\alpha^{*p}_S \frac{1}{|A\setminus S|}.$$\newline

\bibliographystyle{elsarticle-harv} 
\singlespacing
\bibliography{cas-refs}





\end{document}